\documentclass[oribibl]{llncs}
\usepackage{fullpage}
\usepackage{graphicx}
\usepackage{subfigure}
\usepackage{stmaryrd}
\usepackage{url}
\usepackage[linesnumbered,ruled,vlined,english]{algorithm2e}
\usepackage{amsfonts}
\usepackage{amsmath}
\usepackage{amssymb}
\usepackage{epstopdf}
\usepackage[normalem]{ulem}
\usepackage[numbers]{natbib}
\usepackage[usenames,dvipsnames]{color}
\usepackage{xcolor}
\usepackage[space]{grffile}
\usepackage{verbatim}
\usepackage{authblk}
\usepackage{natbib}
\usepackage{booktabs}
\usepackage{array}
\newcommand{\cP}{\mathcal P}
\newcommand{\cW}{\mathcal W}

\newcommand{\cL}{\mathcal L}

\newcommand{\cF}{\mathcal F}

\newcommand{\cX}{\mathcal X}
\newcommand{\cD}{\mathcal D}

\newcommand{\cV}{\mathcal V}

\newcommand{\scLCA}{\textsc{LCA}}
\newcommand{\scAMa}{\textsc{allMAAFs}$^1$}

\newcommand{\scPCa}{\textsc{processCherries}$^1$}
\newcommand{\scAMc}{\textsc{allMAAFs}$^2$}
\def\fc{$\ensuremath{\wedge}$}
\def\fcL{\ensuremath{\bigwedge}}
\def\HiLi{\leavevmode\rlap{\hbox to \hsize{\color{yellow!50}\leaders\hrule height .8\baselineskip depth .5ex\hfill}}}
\title{Fast computation of all maximum acyclic agreement forests for two rooted binary phylogenetic trees}

\author{Benjamin Albrecht
\thanks{
Benjamin Albrecht\\
Institut f\"ur Informatik, Ludwig-Maximilians-Universit\"at, Germany\\
Tel.: +49-89-2180-4069\\
E-mail: \email{albrecht@bio.ifi.lmu.de}}
}

\institute{Institut f\"ur Informatik, Ludwig-Maximilians-Universit\"at, Germany}
\date{\today}

\begin{document}
\maketitle
\vspace{3cm}

\begin{abstract}
Evolutionary scenarios displaying reticulation events are often represented by rooted phylogenetic networks. Due to biological reasons, those events occur very rarely, and, thus, networks containing a minimum number of such events, so-called minimum hybridization networks, are of particular interest for research. Moreover, to study reticulate evolution, biologist need not only a subset but \emph{all} of those networks. To achieve this goal, the less complex concept of rooted phylogenetic trees can be used as building block. Here, as a first important step, the trees are disjoint into common parts, so-called maximum acyclic agreement forests, which can then be turned into minimum hybridization networks by applying further network building algorithms. In this paper, we present two modifications of the first non-naive algorithm --- called \textsc{allMAAFs} --- computing \emph{all} maximum acyclic agreement forests for two rooted binary phylogenetic trees on the same set of taxa. By a simulation study, we indicate that through these modifications the algorithm is on average 8 times faster than the original algorithm making this algorithm accessible to larger input trees and, thus, to a wider range of biological problems.\\


\end{abstract}

\newpage
\section{Introduction}

Phylogenetic studies deal with the reconstruction of evolutionary histories, which are often represented by rooted phylogenetic trees in which each node represents a certain speciation event. The representation of reticulation events combining different evolutionary histories, however, requires the more general concept of phylogenetic networks, which are more complex structures that can additionally consist of nodes of in-degree larger than one. Hybridization, for instance, is an example for such a reticulation event merging a sizable percentage of two genomes of two different species. Now, due to those reticulation events, for a set of species there can exist incongruent gene trees, which are phylogenetic trees based on different genes sampled for this  particular set of species all representing different evolutionary histories. On should, however, keep in mind that there can also exist other biological effects leading to incongruent gene trees, e.g. incomplete lineage sorting, still allowing to model evolutionary histories by phylogenetic trees. 

As it is the case for hybridization events, a reticulation event is in general expected to occur very rarely. Thus, in order to study reticulate evolution, one is interested in the minimum number of such events explaining a set of incongruent gene trees. A further more complicated but strongly connected problem is the computation of \emph{rooted phylogenetic networks} displaying evolutionary scenarios containing such events. In such a network each node with in-degree one reflects a putative speciation event, whereas each node of in-degree of at least two, a so-called reticulation node (or, in respect to hybridization, hybridization node) represents a putative reticulation event.

Phylogenetic studies are often joint works between biologists, mathematicians and computer scientists; usually biologists produce certain data, mathematicians develop a model for analyzing this data and computer scientists implement programs, which enable the investigation of the data based on these models. Regarding the investigation of reticulate evolution, the underlying problem is in general highly combinatorial, which complicates the work of computer scientists as naive approaches are typically insufficient for analyzing biologically relevant data. 

In this work, we tackle this problem by focusing on the first non-naive algorithm \textsc{allMAAFs} \cite{Scornavacca2012} that enables the computation of \emph{all} maximum acyclic agreement forests for two rooted binary phylogenetic trees, which can be considered, from a mathematical point of view, as an intermediate step in calculating \emph{all} possible rooted phylogenetic networks. This is the case, since those networks can then be computed by gluing the components of such maximum acyclic agreement forests back together in a specific way which, for example, has already been demonstrated by the algorithm \textsc{HybridPhylogeny} \cite{Baroni2006}. The computation of a maximum acyclic agreement forest for two binary phylogenetic trees, however, is a well-known \emph{NP-hard} problem \cite{Bordewich-2007/1} and, thus, in order to apply this step to large input trees and thereby making the algorithm accessible to a wider range of biological problems, we have worked out two modifications significantly improving the practical running time of the algorithm \textsc{allMAAFs}.

Notice that so far there have been developed several approaches for analyzing reticulation events \cite{Albrecht2011, Chen2010, Collins2011, Whidden2009}. Whereas some of these approaches just enable the computation of hybridization numbers, some additionally can be used to generate a set of minimum hybridization networks each representing a particular evolutionary scenario. Notice, however, that until now none of these methods is able to calculate \emph{all} minimum hybridization networks, which, as already emphasized, under a biological point of view, is an important feature for studying reticulate evolution.

The paper is organized as follows. First, we give all basic notations and some well-known results of phylogenetics research, which are used throughout this work. Next, we describe two modifications \scAMa~and \scAMc~of the algorithm \textsc{allMAAFs} by presenting its respective pseudo code. In a subsequent step, the correctness of each of these modified algorithms is established by formal proofs and, finally, in order to indicate the speedup obtained from applying these modifications, we present the results of a simulation study comparing runtimes attained from its respective implementations. Lastly, we finish this paper by briefly discussing the theoretical worst-case running time of both algorithms \scAMa~and \scAMc.

\section{Preliminaries}
\label{sec-pre}

In this section, we give all formal definitions that are used throughout this work for describing and discussing the original algorithm \textsc{allMAAFs} as well as our two modified algorithms \scAMa~and \scAMc. These definitions correspond to those given in the work of \citet{Huson2011} and \citet{Scornavacca2012}. We assume here that the reader is familiar with general graph-theoretic concepts.\\

{\bf Phylogenetic trees.} A \emph{rooted phylogenetic $\cX$-tree} $T$ is a directed acyclic connected graph whose edges are directed from the root to the leaves as defined in the following. There is exactly one node of in-degree $0$, namely the \emph{root} of $T$. The set of nodes of out-degree $0$ is called the \emph{leaf set of $T$} and is labeled one-to-one by a \emph{taxa set} $\cX$, also denoted by $\cL(T)$. Here, the taxa set $\cX$ usually consists of certain species or genes whose evolution is outlined by $T$. The tree $T$ is called \emph{binary} if all of its nodes, except the root, provide an in-degree of $1$ and if all of its nodes, except all leaves (the so-called \emph{inner or internal nodes}) provide an out-degree of $2$. Given a node $v$, the label set $\cL(v)$ refers to all taxa that are contained in the subtree rooted at $v$. In addition, given a set of rooted phylogenetic $\cX$-trees $\cF$, by $\cL(\cF)$ we refer to the union of each leaf set of each tree in $\cF$.

Now, given a set of taxa $\cX' \subseteq \cL$, the notation $T(\cX')$ refers to the minimal connected subgraph of $T$ containing $\cX'$. A \emph{restriction} of $T$ to $\cX'$, shortly denoted by $T|_{\cX'}$, is a rooted phylogenetic tree that is obtained from $T(\cX')$ by suppressing each node of both in- and out-degree $1$. Moreover, given two phylogenetic trees $T_1$ and $T_2$ in which $\cX_1$ and $\cX_2$, with $\cX_1\subseteq\cX_2$, denotes the taxa set of $T_1$ and $T_2$, respectively, we say $T_1$ is \emph{contained} in $T_2$, shortly denoted by $T_1\subseteq T_2$, if $T_2|_{\cX_1}$ is isomorphic to $T_1$.

Next, the \emph{lowest common ancestor} of a phylogenetic $\cX$-tree $T$  in terms of a taxa set $\cX' \subseteq \cX$ is the node $v$ in $T$ with $\cX'\subseteq\cL(v)$ such that there does not exist another node $v'$ in $T$ with $\cX'\subseteq\cL(v')$ and $\cL(v')\subseteq\cL(v)$. In the following, such a node is shortly denoted by $\scLCA_T(\cX')$.

Lastly, given a set of phylogenetic trees $\cF$ as well as an edge set $E'$ that is contained in $\cF$ such that for each pair of edges $e_1,e_2\in E'$, with $e_1\neq e_2$, $e_1$ is not adjacent to $e_2$. Then, by $\cF-E'$ we refer to the set $\cF'$ of trees that is obtained from $\cF$ by \emph{cutting $E'$}. More precisely, first each edge in $E'$ is deleted in $\cF$ and then each node of both in- and out-degree $1$ is suppressed. Notice that by deleting an edge of a tree $F$ in $\cF$, this tree is separated into two parts $F_a$ and $F_b$ so that the resulting forest equals $\cF\setminus\{F\}\cup\{F_a,F_b\}$. Consequently, after deleting $E'$ from $\cF$ the resulting set $\cF'$ contains precisely $|\cF|+|E'|$ trees.

During the algorithm \textsc{allMAAFs}, an agreement forests for two rooted binary phylogenetic $\cX$-trees $T_1$ and $T_2$ is calculated by cutting down one of both trees, say $T_2$, into several components such that each component corresponds to a restricted subtree of $T_1$. In order to keep track of the component $F_{\rho}$ containing the root, the root of $T_2$ has to be a node that is marked in a specific way. Thus, throughout this paper, we regard the root of each rooted binary phylogenetic $\cX$-tree as a node that is attached to its original root and to a taxon $\rho\not\in \cX$ (cf.~Fig.~\ref{fig-example}(a)).\\

\textbf{Phylogenetic networks.} A \emph{rooted phylogenetic network} $N$ on $\cX$ is a rooted connected digraph whose edges are directed from the root to the leaves as defined in the following. There is exactly one node of in-degree $0$, namely the \emph{root}, and no nodes of both in- and out-degree $1$. The set of nodes of out-degree $0$ is called the \emph{leaf set of $N$} and is labeled one-to-one by the \emph{taxa set} $\cX$, also denoted by $\cL(N)$. In contrast to a phylogenetic tree, such a network may contain undirected but not any directed cycles. Consequently, $N$ can contain nodes of in-degree larger than or equal to $2$, which are called \emph{reticulation nodes} or \emph{hybridization nodes}. Moreover, all edges that are directed into such a reticulation node are called \emph{reticulation edges} or \emph{hybridization edges}.\\

{\bf Hybridization Networks.} A hybridization network $N$ for two rooted binary phylogenetic $\cX$-trees $T_1$ and $T_2$ is a rooted phylogenetic network on $\cX$ \emph{displaying} $T_1$ and $T_2$. More precisely, this means that for each tree $T\in\{T_1,T_2\}$ there exists a set $E'\subseteq E(N)$ of reticulation edges such that $T$ can be derived from $N$ by conducting the following steps. 
\begin{enumerate}
\item[(1)] Delete each edge from $N$ that is not contained in $E'$.
\item[(2)] Remove each node whose corresponding taxon is not contained in $\cX'$.
\item[(3)] Remove each unlabeled node of out-degree $0$ repeatedly.
\item[(4)] Suppress each node of both in- and out-degree $1$.  
\end{enumerate}
Notice that, for determine $E'$, it suffices to select at most one in-edge of each hybridization node in $N$.

From a biological point of view, this means that $N$ displays $T$ if each speciation event of $T$ is reflected by $N$. Moreover, each internal node of in-degree~$1$ represents a speciation event and each internal node providing an in-degree of at least~$2$ represents a reticulation event or, in terms of hybridization, a hybridization event. This means, in particular, that such a latter node represents an individual whose genome is a \emph{chimaera} of several parents. Thus, such a node $v$ of in-degree larger than or equal to $2$ is called a \emph{hybridization node} and each edge directed into $v$ is called a \emph{hybridization edge}.

Now, based on those hybridization nodes, the \emph{reticulation number} $r(N)$ of a hybridization network $N$ is defined by 
\begin{equation}
r(N)=\sum_{v\in V:\delta^-(v)>0}\left(\delta^-(v)-1\right)=|E|-|V|+1,
\label{04-eq-retNumBin}
\end{equation}
where $V$ denotes the node set and $E$ the edge set of $N$. Next, based on the definition of the reticulation number, for two rooted binary phylogenetic $\cX$-trees $T_1$ and $T_2$ the hybridization number $h(\{T_1,T_2\})$ is 
\begin{equation}
\text{min}\{r(N):\text{N is a hybridization network displaying }T_1\text{ and }T_2\}.
\label{04-eq-hNumBin}
\end{equation}
Lastly, we call a \emph{hybridization network} $N$ for two rooted binary phylogenetic $\cX$-trees $T_1$ and $T_2$ a \emph{minimum hybridization network} if $r(N)=h(\{T_1,T_2\})$.\\

{\bf Forests.} Let $T$ be a rooted nonbinary phylogenetic $\cX$-tree $T$. Then, we call any set of rooted nonbinary phylogenetic trees $\cF=\{F_1,\dots,F_k\}$ with $\cL(\cF)=\cX$ a \emph{forest on $\cX$}, if we have for each pair of trees $F_i$ and $F_j$ that $\cL(F_i)\cap\cL(F_j)=\emptyset$. Moreover, if additionally for each component $F$ in $\cF$ the tree $T|_{\cL(F)}$ equals $F$, we say that $\cF$ is a \emph{forest for $T$}. Lastly, let $\cF$ be a forest for a rooted binary phylogenetic $\cX$-tree $T$. Then, by $\overline \cF$ we refer to the forest that is obtained from $\cF$ by deleting each component only consisting of an isolated node as well as the element containing the node labeled by taxon $\rho$ if it contains at most one edge.\\

{\bf Agreement Forests.} For technical purpose, the definition of agreement forests is based on two rooted binary phylogenetic $\cX$-trees $T_1$ and $T_2$ whose roots are marked by a unique taxon $\rho\not\in\cX$ as follows. Let $r_i$ be the root of the tree $T_i$ with $i\in\{1,2\}$. Then, we first create a new node $v_i$ labeled by a new taxon $\rho\not\in\cX$ and then attach this node to $r_i$ by inserting the edge $(v_i,r_i)$. Notice that this case $v_1$ and $v_2$ is the new root of $T_1$ and $T_2$, respectively. Moreover, since we consider $\rho$ as being a new taxon, the taxa set of both trees is $\cX\cup\{\rho\}$.

Now, given two such marked rooted binary phylogenetic trees $T_1$ and $T_2$ on $\cX\cup\rho$, a set of components $\cF=\{F_{\rho},F_1,\dots,F_k\}$ is an \emph{agreement forest} for $T_1$ and $T_2$ if the following three conditions are satisfied.

\begin{enumerate} 
\item[(1)] Each component $F_i$ with taxa set $\cX_i$ equals $T_1|_{\cX_i}$ and $T_2|_{\cX_i}$. 
\item[(2)] There is exactly one component, denoted as $F_{\rho}$, with $\rho\in\cL(F_{\rho})$.
\item[(3)] Let $\cX_{\rho},\cX_1,\dots,\cX_{k}$ be the respective taxa sets of $F_{\rho},F_1,\dots,F_{k}$. All trees in $\{T_1(\cX_i)|i\in\{\rho,1,\dots,k\}\}$ as well as $\{T_2(\cX_i)|i\in\{\rho,1,\dots,k\}\}$ are node disjoint subtrees of $T_1$ and $T_2$, respectively.
\end{enumerate} 
An illustration of an agreement forest is given in Figure~\ref{fig-example}(b).

Lastly, an agreement forest for two rooted binary phylogenetic $\cX$-trees containing a minimum number of components is called a \emph{maximum agreement forest}. Now, based on maximum agreement forests, one can compute a minimum hybridization network if an additional constraint is satisfied, which is presented in the following. Notice that, from a biological point of view, this constraint is necessary since it prevents species from inheriting genetic material from their own offspring.\\

{\bf Modified ancestor descendant graphs.} Given two rooted binary phylogenetic $\cX$-trees $T_1$ and $T_2$ together with a forest $\cF=\{F_{\rho},F_1,\dots,F_k\}$ for $T_1$ (or $T_2$), the modified ancestor descendant graph $AG^*(T_1,T_2,\cF)$ with node set $\cF$ contains a directed edge $(F_i,F_j)$, $i\ne j$, if
\begin{enumerate} 
\item[(1)] the root of $T_1(\cL(F_i))$ is an ancestor of the root of $T_1(\cL(F_j))$,  
\item[(2)] or the root of $T_2(\cL(F_i))$ is an ancestor of the root of $T_2(\cL(F_j))$.
\end{enumerate}
An illustration of such an modified ancestor descendant graph is given in Figure~\ref{fig-example}(c).

Now, we say that $\cF$ is an \emph{acyclic agreement forest} if $\cF$ satisfies the conditions of an agreement forest and if the corresponding graph $AG^*(T_1,T_2,\cF)$ does not contain any \emph{directed} cycles. Again, each acyclic agreement forest of minimum size among all acyclic agreement forests is called a \emph{maximum acyclic agreement forest}. Note that the definition of an ancestor descendant graph given in both works \citet{Baroni2005} as well as \citet{Scornavacca2012} is different to the one presented above. Regarding the definitions of those two works, the set of components $\cF$ has to be an agreement forest for $T_1$ and $T_2$ and, thus, this definition is more strict than the one given for $AG^*(T_1,T_2,\cF)$ in which $\cF$ just has to be a forest for $T_1$ (or $T_2$). For a better illustration, in Figure~\ref{fig-example} we give an example of an acyclic agreement forest together with its underlying modified ancestor descendant graph.

\begin{figure}[t]
\centering
\includegraphics[scale=0.75]{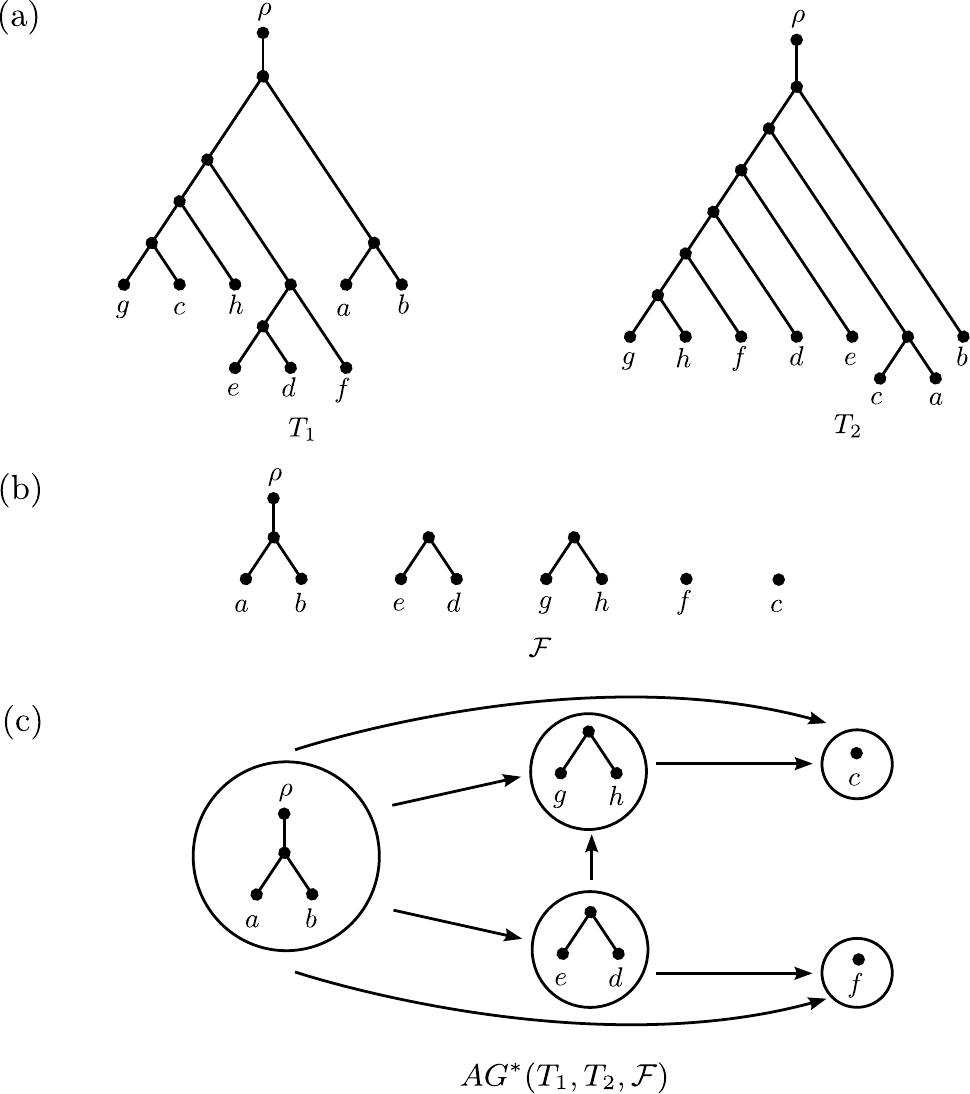}
\caption{(a) Two rooted binary phylogenetic $\cX$-trees $T_1$ and $T_2$ with taxa set $\cX=\{a,b,c,d,e,f,g,h,\rho\}$. (b) An acyclic agreement forest $\cF$ for $T_1$ and $T_2$. (c) The directed graph $AG^*(T_1,T_2,\cF)$ not containing any directed cycles and, thus, $\cF$ is acyclic. }
\label{fig-example}
\end{figure} 

Notice that the concept of maximum acyclic agreement forests for two rooted binary phylogenetic $\cX$-trees is of interest, since its size minus one corresponds to the hybridization number of these two trees, which is stated by Theorem \ref{theo-h} given in the paper of \citet{Baroni2005}. 

\begin{theorem}[\cite{Baroni2005}]
\label{theo-h}
Let $\cF$ be a maximum acyclic agreement forest of size $k$ for two rooted binary phylogenetic $\cX$-trees $T_1$ and $T_2$, then $h(T_1,T_2)=k$.
\end{theorem}

{\bf Cherries.} For a rooted binary phylogenetic $\cX$-tree $T$ we call two of its leaves $a$~and~$c$ a \emph{cherry}, denoted by $\{\cL(a),\cL(c)\}$, if both nodes have the same parent. Moreover, let $R$ and $S$ be two rooted binary phylogenetic trees on $\cX_R$ and $\cX_S$, respectively, so that $\cX_R\subseteq\cX_S$, and let $\cF$ be a forest for $S$. Then, a cherry $\{\cL(a),\cL(c)\}$ in $R$ is called a \emph{common cherry of $R$ and $\cF$}, if there exists a cherry $\{\cL(a'),\cL(c')\}$ in $\cF$ with $\cL(a')=\cL(a)$ and $\cL(c')=\cL(c)$. Otherwise, the cherry $\{\cL(a),\cL(c)\}$ in $R$ is called \emph{contradicting cherry of $R$ and $\cF$}. Moreover, in order to ease reading, if there exists a component in $\cF$ containing two leaves labeled by $\cL(a)$ and $\cL(c)$, respectively, we write $a\sim_{\cF}c$. Otherwise, if such a component does not exist, we write $a\not\sim_{\cF}c$.

Moreover, let $\cF$ be a forest for a rooted binary phylogenetic $\cX$-tree and let $\{\cL(a),\cL(c)\}$ be a cherry that is contained in $\cF$ in which $e$ denotes the in-edge of its parent $p$. Then, if $e$ exists, by $\cF\div\{\cL(a),\cL(c)\}$ we simply refer to the forest $\cF-\{e\}$ and, otherwise, if $p$ has in-degree $0$, to $\cF$. Furthermore, let $\cP$ be the path connecting $a$ and $c$ in $\cF$. Then, the \emph{set of pendant edges for $\{\cL(a),\cL(c)\}$} contains each edge $(v,w)$ with $v\in V(\cP)\setminus\{a,c\}$ and $w\not\in V(\cP)$, where $V(\cP)$ denotes the set of nodes in $\cP$. Notice that in general there exist several, namely precisely $|V(\cP)|-3$, edges satisfying the condition of such an edge.\\

{\bf Cherry Reductions.} Let $\cF$ be a forest for a rooted binary phylogenetic $\cX$-tree and let $\{\cL(a),\cL(c)\}$ be a cherry of a component $F_i$ in $\cF$. Then, a \emph{cherry reduction}, according to a cherry $\{\cL(a),\cL(c)\}$ in one of its components $F_i$, implies the following two operations.

\begin{enumerate} 
\item[(1)] The parent of the two nodes $a$ and $c$ is labeled by $\{\cL(a),\cL(c)\}$.
\item[(2)] Both nodes $a$ and $c$ together with their adjacent edges are deleted from $F_i$. 
\end{enumerate}

Throughout the algorithm, such a reduction step of a common cherry $\{\cL(a),\cL(c)\}$ in $\cF$ is shortly denoted by $\cF[\{\cL(a),\cL(c)\}\rightarrow \cL(a)\cup\cL(c)]$. Furthermore, the reverse notation $\cF[\cL(a)\cup\cL(c)\rightarrow \{\cL(a),\cL(c)\}]$ describes the insertion of two new taxa labeled by $\cL(a)$ and $\cL(c)$, respectively, together with removing the label $\cL(a)\cup\cL(c)$ from its parent. Notice that, for applying such insertion steps, one has to keep track of the preceding reduction steps.

Equivalently, for a cherry $\{\cL(a),\cL(c)\}$ in a rooted binary phylogenetic $\cX$-tree $R$ we write $R[\{\cL(a),\cL(c)\}\rightarrow \cL(a)\cup\cL(c)]$ to denote a cherry reduction of $\{\cL(a),\cL(c)\}$ in $R$.

\section{The Algorithm \textsc{allMAAFs}}
\label{14-sec-allMAAFs}

In this section, we present the algorithm \textsc{allMAAFs} that was first published in the work of Scornavacca~\emph{et al.}~\citet{Scornavacca2012}. To increase its readability, we decided to split the original algorithm \textsc{allMAAFs} into six parts (cf.~Alg.~\ref{14-alg-main}--\ref{14-alg-contra_cherry}). Note that, apart from its graphical representation, our presentation of the algorithm \textsc{allMAAFs} together with its terminology in general adheres to the original algorithm. 

\begin{algorithm}[]
\scriptsize
\KwData{Two rooted binary phylogenetic $\cX$-trees $S$ and $T$, a rooted binary phylogenetic tree $R$ and a forest $\cF$ such that $\cL(R)=\cL(\overline{\cF})$ and $\cL(T)=\cL(\cF)$, an integer $k$, and a list $M$ that contains information of previously reduced cherries.}

\KwResult{A set $\boldsymbol{\cF}$ of forests for $\cF$ and an integer. In particular, if $\cF=T$, $R=S$, $M=\emptyset$, and $k\geq h(S,T)$ is the input to {\sc allMAAFs}, the output precisely consists of all maximum-acyclic-agreement forests for $S$ and $T$ and their respective hybridization number.}

\If{$k<0$}{\Return $(\emptyset,k-1)$\;} 
\If{$|\cL(R)|=0$}{
	$\cF' \gets$ \textsc{cherryExpansion}($\cF$, $M$)\;
	\If {$AG(S,T,\cF')$ \textnormal{is acyclic}}{
		\Return ($\cF'$, $|\cF'|-1$)\;
	}
	\Else{
		\Return $(\emptyset,k-1)$\;
	}
}
\Else{
	let $\{\cL(a),\cL(c)\}$ be a cherry of $R$\;
	\If{$\{\cL(a),\cL(c)\}$ \textnormal{ is a common cherry of} $R$ \textnormal{and}
$\cF$}{
		\Return (\textsc{ProcessCommonCherry}($S$, $T$, $R$, $\cF$, $k$,
$M$, $\{\cL(a),\cL(c)\}$))\;
	}
	\If{$k\neq ( |\cF|-1)$ \textnormal{or} $\{\cL(a),\cL(c)\}$ \textnormal{ is a
contradicting cherry of} $R$ \textnormal{and} $\cF$ }{
		\Return (\textsc{ProcessContradictingCherry}($S$, $T$, $R$,
$\cF$, $k$, $M$, $\{\cL(a),\cL(c)\}$))\;
	}
}
\caption{\textsc{allMAAFs}($S$, $T$, $R$, $\cF$, $k$, $M$)}
\label{14-alg-main}
\end{algorithm}

\begin{algorithm}[]
\scriptsize
\While{M \em{is not empty}}{
 $M \leftarrow \text{remove last element of M, say} \{\cL(a),\cL(c)\}$\;
 $\cF \leftarrow \cF[\cL(a)\cup\cL(c) \rightarrow \{\cL(a),\cL(c)\}]$\;
}
\Return{$\cF$}
\caption{ \textsc{cherryExpansion}$(\cF,M)$}
\label{14-alg-exp}
\end{algorithm}

\begin{algorithm}[]
\begin{footnotesize}
$M'\gets \mbox{Add } \{\cL(a), \cL(c)\}\mbox{ as last element of } M$\;
 %last element of M important if the reduced subtree contains other dummy taxa that are not part of the original tree S
 $R' \gets R[\{\cL(a),\cL(c)\}\rightarrow \cL(a)\cup \cL(c)]$\;
	%$F_i\gets F_i \in \cF'$ and $A\subseteq \cL(F_i)$\;
$\cF'\gets\cF[\{\cL(a),\cL(c)\}\rightarrow \cL(a)\cup \cL(c) ]$\;
\Return $(R',\cF',M')$
\caption{ \textsc{cherryReduction}$(R,\cF,M,\{\cL(a),\cL(c)\})$\label{a:subRed}} 
\end{footnotesize}
\end{algorithm}

\begin{algorithm}[]
\scriptsize
$(R',\cF',M') \gets$ \textsc{cherryReduction}($R$, $\cF$, $M$, $\{\cL(a),\cL(c)\}$)\;
($\boldsymbol{\cF_r}$, $k_r$) $\gets$ \textsc{allMAAFs}($S$, $T$,
$R'|_{\cL(\overline{\cF'})}$, $\cF'$, $k$, $M'$)\;
\If{$\boldsymbol{\cF_r}\neq \emptyset$}{
	$k\gets \min(k, k_r)$\;
}
\lIf{$(k = |\cF|-1)$}{
	\Return ($\boldsymbol{\cF_r}$, $k$)\;
}
\Else{
	($\boldsymbol{\cF_a}$, $k_a$) $\gets$ \textsc{allMAAFs}($S$, $T$, $R|_{\cL(\overline{\cF-\{e_a\}})}$, $\cF-\{e_a\}$, $k-1$, $M$)\;
	\If{$\boldsymbol{\cF_a}\neq \emptyset$}{$k \gets \min(k, k_a-1)$\;}
	\lIf{$(k_a -1 = k)$}{
		$\boldsymbol{\cF} \leftarrow \boldsymbol{\cF} \cup \boldsymbol{\cF_a}$\;
	}
	($\boldsymbol{\cF_c}$, $k_c$) $\gets$ \textsc{allMAAFs}($S$, $T$, $R|_{\cL(\overline{\cF-\{e_c\}})}$, $\cF-\{e_c\}$, $k-1$, $M$)\;
	\If{$\boldsymbol{\cF_c}\neq \emptyset$}{
		$k \gets \min(k,k_c -1)$\;
	}
	$\boldsymbol{\cF} \gets \emptyset$\;
	\lIf{$(k_a -1 = k)$}{
		$\boldsymbol{\cF} \leftarrow \boldsymbol{\cF_a}$\;
	}
	\lIf{$(k_c -1 = k)$}{
		$\boldsymbol{\cF} \leftarrow \boldsymbol{\cF} \cup \boldsymbol{\cF_c}$\;
	}
	\lIf{$(k_r = k)$}{
		$\boldsymbol{\cF} \leftarrow \boldsymbol{\cF} \cup \boldsymbol{\cF_r}$\;
	}
	\Return ($\boldsymbol{\cF}$, $k$)\;
}
\caption{ \textsc{ProcessCommonCherry}($S$, $T$, $R$, $\cF$, $k$, $M$,
$\{\cL(a),\cL(c)\}$)}
\label{14-alg-common_cherry}
\end{algorithm}

\begin{algorithm}[]
\scriptsize
($\boldsymbol{\cF_a}$, $k_a$) $\gets$ \textsc{allMAAFs}($S$, $T$,
$R|_{\cL(\overline{\cF-\{e_a\}})}$, $\cF-\{e_a\}$, $k-1$, $M$)\;
\If{$\boldsymbol{\cF_a}\neq \emptyset$}{$k \gets \min(k, k_a-1)$\;}
($\boldsymbol{\cF_c}$, $k_c$) $\gets$ \textsc{allMAAFs}($S$, $T$,
$R|_{\cL(\overline{\cF-\{e_c\}})}$, $\cF-\{e_c\}$, $k-1$, $M$)\;
\If{$\boldsymbol{\cF_c}\neq \emptyset$}{$k \gets \min(k,k_c -1)$\;}
$\boldsymbol{\cF} \gets \emptyset$\;
\If{$a \nsim_{\cF} c$ }{				 
	\lIf{$(k_a -1= k)$}{
		$\boldsymbol{\cF} \leftarrow \boldsymbol{\cF_a}$\;
	}
	\lIf{$(k_c -1= k)$}{
		$\boldsymbol{\cF} \leftarrow \boldsymbol{\cF} \cup \boldsymbol{\cF_c}$\;
	}
	\Return ($\boldsymbol{\cF}$, $k$)\;
}
\Else{
	($\boldsymbol{\cF_B}$, $k_B$) $\gets$ \textsc{allMAAFs}($S$, $T$,
$R|_{\cL(\overline{\cF-\{e_B\}})}$, $\cF-\{e_B\}$, $k-1$, $M$)\;		
	
	\If{$\boldsymbol{\cF_B}\neq \emptyset$}{
		$k \gets \min(k, k_B-1 )$\;
	} 
	\lIf{$(k_a -1 = k)$}{
		$\boldsymbol{\cF} \leftarrow \boldsymbol{\cF_a}$\;
	}
	\lIf{$(k_B -1 = k)$}{
		$\boldsymbol{\cF} \leftarrow \boldsymbol{\cF} \cup \boldsymbol{\cF_B}$\;
	}
	\lIf{$(k_c -1 = k)$}{
		$\boldsymbol{\cF} \leftarrow \boldsymbol{\cF} \cup \boldsymbol{\cF_c}$\;
	}
	\Return ($\boldsymbol{\cF}$, $k$)\;
}
\caption{ \textsc{ProcessContradictingCherry}($S$, $T$, $R$, $\cF$, $k$, $M$,
$\{\cL(a),\cL(c)\}$)}
\label{14-alg-contra_cherry}
\end{algorithm}

\section{Modifications to \textsc{allMAAFs}}
\label{14-sec-mod}

In this section, we present two modifications of the algorithm \textsc{allMAAFs} that, on the one hand, do not improve its theoretical runtime but, one the other hand, significantly improve its practical running time, which will be indicated by a simulation study reported in Section~\ref{14-sec-simulation}. 

The first modification improves the processing of a certain type of contradicting cherry whereas for the other two modifications only the processing step of a common cherry is of interest.
 
Given a contradicting cherry $\{\cL(a),\cL(c)\}$ for $R$ and $\cF$ with $a\sim_{\cF}c$, the original algorithm conducts three recursive calls. One by recursively calling the algorithm with $\cF-\{e_a\}$ and $R|_{\overline{\cF-\{e_a\}}}$, one by recursively calling the algorithm with $\cF-\{e_c\}$ and $R|_{\overline{\cF-\{e_c\}}}$ and one by recursively calling the algorithm with $\cF-\{e_B\}$ and $R|_{\overline{\cF-\{e_B\}}}$, in which $e_a$ and $e_c$ refers to the in-edge of leaf $a$ and $c$, respectively, in $\cF$ and $e_B$ refers to an in-edge of a subtree lying on the path connecting $a$ and $c$ in $\cF$. Regarding the latter recursive call, in the upcoming part of this work we will show that, in order to compute all maximum acyclic agreement forests, instead of cutting just one in-edge $e_B$ one can cut all of those in-edges all at once. 

Moreover, given a common cherry $\{\cL(a),\cL(c)\}$ for $R$ and $\cF$, the original algorithm \textsc{allMAAFs} always branches into three new computational paths; one path corresponding to the cherry reduction of $\{\cL(a),\cL(c)\}$ and two corresponding to the deletion of both in-edges of the two leaves $a$ and $c$ (cf.~Alg.~\ref{14-alg-common_cherry}, line 7--15). To understand the sense of our two modifications, one has to take the necessity of these two additional edge deletions into account. 

Therefor, we demonstrate a specific scenario that is outlined in Figure~\ref{14-fig-example} showing two phylogenetic trees $T_1$ and $T_2$ as well as a maximum acyclic agreement $\cF=\{((a,b),\rho),(e,d),(g,h),c,f\}$ for those two trees. By running the algorithm \textsc{allMAAFs} for $T_1$ and $T_2$ without deleting in-edges of a common cherry, the given maximum acyclic agreement forest $\cF$ is never computed. This is due to the fact that, once the component $F_i=((g,h),f)$ occurs on any computational path, $F_i$ will be part of the resulting agreement forest. Because of $F_i$ and the component $(e,d)$, such a resulting agreement forest is never acyclic and, thus, does not satisfy the conditions of an acyclic agreement forest. However, by cutting instead of contracting the common cherry $\{(g,h),f\}$, the resulting agreement forest turns into the maximum acyclic agreement forest $\cF$. This example implies that sometimes the deletion of in-edges corresponding to taxa of a common cherry is necessary, which is, however, in practice not often the case, and, thus, the original algorithm \textsc{allMAAFs} usually produces a lot of additional unnecessary computational steps compared with our second modification offering a different solution for such a scenario. 

\begin{figure}
\centering
\includegraphics[scale=0.75]{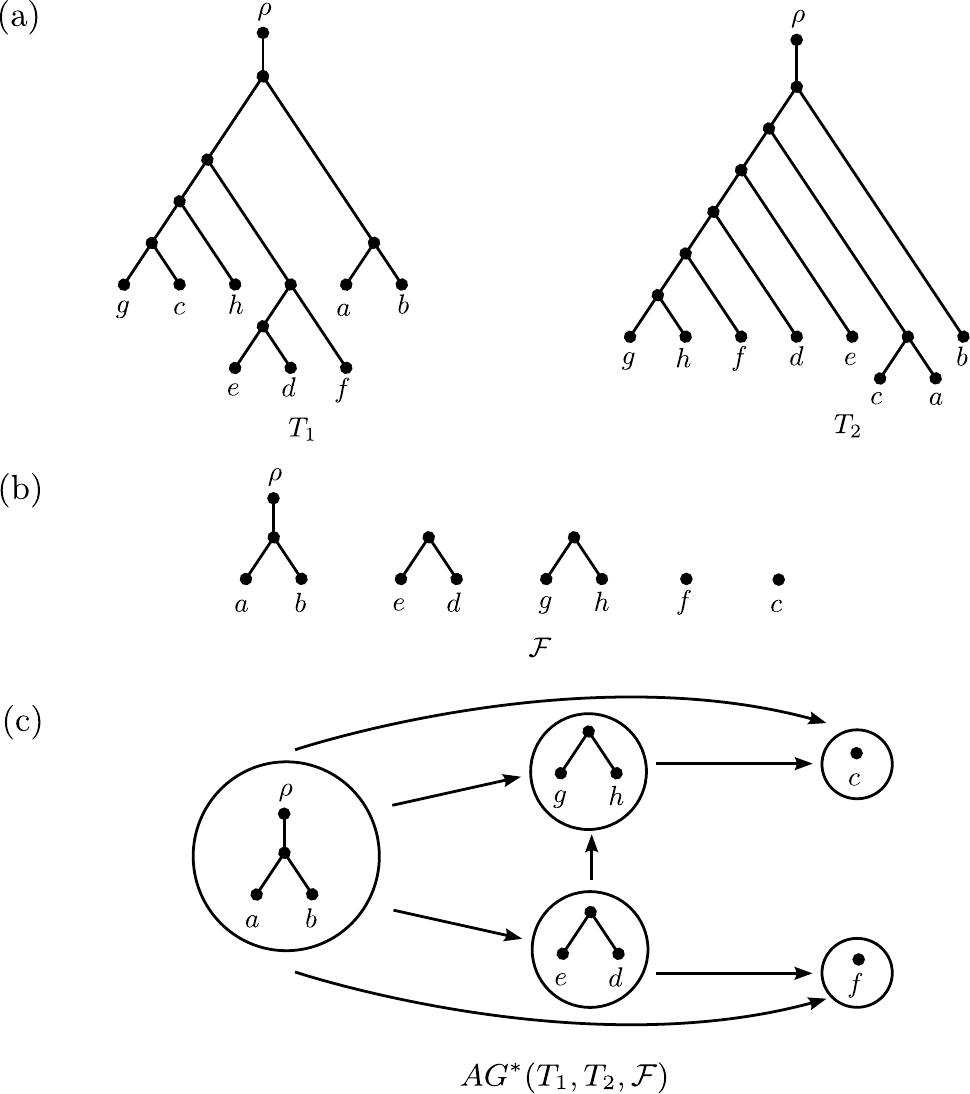}
\caption{(a) Two rooted binary phylogenetic $\cX$-trees $T_1$ and $T_2$ with taxa set $\cX=\{a,b,c,d,e,f,g,h,\rho\}$. (b) An acyclic agreement forest $\cF$ for $T_1$ and $T_2$. (c) The directed graph $AG^*(T_1,T_2,\cF)$ not containing any directed cycles and, thus, $\cF$ is acyclic. }
\label{14-fig-example}
\end{figure} 

\subsection{The Algorithm \scAMa}
\label{sec-mod3}

Our first algorithm \scAMa~is a modification of the original algorithm \textsc{allMAAFs} improving the processing of contradicting cherries. Let $\{\cL(a),\cL(c)\}$ be a contradicting cherry of $R$ and $\cF$ such that $a\sim_{\cF}c$ holds and let $e_B$ be an edge that is defined as follows. Let $\cP$ be the path connecting $a$ and $c$ in $ \cF$. Then, the edge set $E_B$ contains each edge $e_B=(v,w)$ with $v\in V(\cP)\setminus\{a,c\}$ and $w\not\in V(\cP)$, where $V(\cP)$ denotes the set of nodes of $\cP$. Note that in this case there exist precisely $|V(\cP)-3|$ edges satisfying the conditions of such an edge $e_B$. Now, if such a cherry $\{\cL(a),\cL(c)\}$ occurs, the original algorithm \textsc{allMAAFs} branches into a computational path by cutting exactly one of those edges in $E_B$. We will show, however, that in this case the whole set $E_B$ can be cut from $\cF$ all at once without having an impact on the computation of all maximum acyclic agreement forests for both input trees. In Algorithm~\ref{14-alg-dagger_contra_cherry}, we give a pseudo code of \scAMa. Note that, for the sake of clarity, we just present the modified part in respect of the original algorithm dealing with the processing of contradicting cherries. The remaining parts are unmodified and can be looked up in Section \ref{14-sec-allMAAFs}.

\begin{algorithm}[]
\scriptsize
($\boldsymbol{\cF_a}$, $k_a$) $\gets$ \scAMa($S$, $T$,
$R|_{\cL(\overline{\cF-\{e_a\}})}$, $\cF-\{e_a\}$, $k-1$, $M$)\;
\If{$\boldsymbol{\cF_a}\neq \emptyset$}{$k \gets \min(k, k_a-1)$\;}
($\boldsymbol{\cF_c}$, $k_c$) $\gets$ \scAMa($S$, $T$,
$R|_{\cL(\overline{\cF-\{e_c\}})}$, $\cF-\{e_c\}$, $k-1$, $M$)\;
\If{$\boldsymbol{\cF_c}\neq \emptyset$}{$k \gets \min(k,k_c -1)$\;}
$\boldsymbol{\cF} \gets \emptyset$\;
\If{$a \nsim_{\cF} c$ }{				 
	\lIf{$(k_a -1= k)$}{
		$\boldsymbol{\cF} \leftarrow \boldsymbol{\cF_a}$\;
	}
	\lIf{$(k_c -1= k)$}{
		$\boldsymbol{\cF} \leftarrow \boldsymbol{\cF} \cup \boldsymbol{\cF_c}$\;
	}
	\Return ($\boldsymbol{\cF}$, $k$)\;
}
\Else{
	\HiLi($\boldsymbol{\cF_B}$, $k_B$) $\gets$ \scAMa($S$, $T$, $R|_{\cL(\overline{\cF-E_B})}$, $\cF-E_B$, $k-|E_B|$, $M$)\;		
	
	\If{$\boldsymbol{\cF_B}\neq \emptyset$}{
		$k \gets \min(k, k_B-1)$\;
	} 
	\lIf{$(k_a -1 = k)$}{
		$\boldsymbol{\cF} \leftarrow \boldsymbol{\cF_a}$\;
	}
	\HiLi\lIf{$(k_B -1 = k)$}{
		$\boldsymbol{\cF} \leftarrow \boldsymbol{\cF} \cup \boldsymbol{\cF_B}$\;
	}
	\lIf{$(k_c -1 = k)$}{
		$\boldsymbol{\cF} \leftarrow \boldsymbol{\cF} \cup \boldsymbol{\cF_c}$\;
	}
	\Return ($\boldsymbol{\cF}$, $k$)\;
}
\caption{ \textsc{ProcessContradictingCherry}$^1$($S$, $T$, $R$, $\cF$, $k$, $M$,
$\{\cL(a),\cL(c)\}$)}
\label{14-alg-dagger_contra_cherry}
\end{algorithm}

\subsection{The Algorithm \scAMc}
\label{14-sec-star}

Our third algorithm \scAMc~is again a modification of our first algorithm \scAMa~and is based on a tool turning agreement forests into acyclic agreement forests. This tool, published by Whidden \emph{et al.}~\citet{Whidden2011}, is based on the concept of an \emph{expanded cycle graph} refining cyclic agreement forest. Due to such additional refinement steps, which are performed right after the computation of each maximum agreement forest, both cutting steps for processing a common cherry can be omitted. The simulation study in Section~\ref{14-sec-simulation} indicates that this refinement step is efficient enough so that this modification in general outperforms the original algorithm \textsc{allMAAFs} as well as our first modification \scAMa. In Algorithm~\ref{14-alg-mod2-main}~and~\ref{14-alg-mod2-common_cherry}, we present a pseudo code describing the algorithm \scAMc. Again, for the sake of clarity, we restrict the presentation to only those parts that are modified in respect to \scAMa. Moreover, we omit a description of the subroutine conducting the refinement step, denoted by \textsc{RefineForest}, as it can be looked up in the work of Whidden \emph{et al.}~\citet{Whidden2011}. Notice that, due to this refinement step, an implementation of this modification is quite more expensive compared with our first modified algorithm \scAMa. 

\begin{algorithm}[]
\scriptsize
\If{$k<0$}{\Return $(\emptyset,k-1)$\;} 
\If{$|\cL(R)|=0$}{
	$\cF' \gets$ \textsc{cherryExpansion}($\cF$, $M$)\;
	\HiLi$\boldsymbol\cF'' \gets$ \textsc{RefineForest}($\cF'$)\;
	\HiLi$\boldsymbol\cF \gets \emptyset$\;
	\HiLi\ForEach{$\cF''\in\boldsymbol\cF''$}{
		\HiLi\If{$|\cF''|=k$}{
			\HiLi$\boldsymbol\cF\gets\boldsymbol\cF\cup\cF''$\;
		}
		\HiLi\ElseIf{$|\cF''|<k$}{
			\HiLi$k\gets|\cF''|$\;
			\HiLi$\boldsymbol\cF\gets\{\cF''\}$\;
		}
	}
	\HiLi\Return $(\boldsymbol\cF,k-1)$\;
}
\Else{
	let $\{\cL(a),\cL(c)\}$ be a cherry of $R$\;
	\If{$\{\cL(a),\cL(c)\}$ \textnormal{ is a common cherry of} $R$ \textnormal{and} $\cF$}{
		\HiLi\Return (\textsc{ProcessCommonCherry}$^3$($S$, $T$, $R$, $\cF$, $k$, $M$, $\{\cL(a),\cL(c)\}$))\;
	}
	\If{$k\neq ( |\cF|-1)$ \textnormal{or} $\{\cL(a),\cL(c)\}$ \textnormal{ is a contradicting cherry of} $R$ \textnormal{and} $\cF$ }{
		\Return (\textsc{ProcessContradictingCherry}($S$, $T$, $R$, $\cF$, $k$, $M$, $\{\cL(a),\cL(c)\}$))\;
	}
}
\caption{\scAMc($S$, $T$, $R$, $\cF$, $k$, $M$)}
\label{14-alg-mod2-main}
\end{algorithm}

\begin{algorithm}[]
\scriptsize
$(R',\cF',M') \gets$ \textsc{cherryReduction}($R$, $\cF$, $M$, $\{\cL(a),\cL(c)\}$)\;
($\boldsymbol{\cF_r}$, $k_r$) $\gets$ \scAMc($S$, $T$,
$R'|_{\cL(\overline{\cF'})}$, $\cF'$, $k$, $M'$)\;
\lIf{$(k_r=k)$}{
\Return ($\boldsymbol{\cF_r}$, $k$)\;
}
\Return ($\emptyset$, $k$)\;
\caption{ \textsc{ProcessCommonCherry}$^3$($S$, $T$, $R$, $\cF$, $k$, $M$,
$\{\cL(a),\cL(c)\}$)}
\label{14-alg-mod2-common_cherry}
\end{algorithm}

\clearpage
\section{Proofs of Correctness}
\label{14-sec-proof}

In this section, we give formal proofs showing the correctness of all two modified algorithms presented in Section~\ref{14-sec-mod}. In a first step, however, we give some further definitions that are crucial for what follows.

\subsection{The algorithm{\sc~processCherries}}

In the following, we introduce the algorithm{\sc~processCherries}, which has already been utilized in the paper of Scornavacca \emph{et al.}~\citet{Scornavacca2012}. This algorithm is a simplified version of the algorithm \textsc{allMAAFs} describing one of its computational paths by a list of \emph{cherry actions}. Notice that this algorithm is a major tool that will help us to establish the correctness of our two modified algorithms.\\

\textbf{Cherry actions.} Let $R$ be a rooted binary phylogenetic $\cX$-tree and let $\cF$ be a forest for $R$. Then, a cherry action $\fc=(\{\cL(a),\cL(c)\},e)$ is a tuple containing a set $\{\cL(a),\cL(c)\}$ of two taxa of two leaves $a$ and $c$ as well as an edge $e$. We say that $\wedge$ is a \emph{cherry action for $R$ and $\cF$}, if $\{\cL(a),\cL(c)\}$ is a cherry of $R$ and if, additionally, one of the following three conditions is satisfied.

\begin{enumerate}
\item[(1)] Either $\{\cL(a),\cL(c)\}$ is a common cherry of $R$ and $\cF$ and $e \in \{\emptyset,e_a,e_c\}$,
\item[(2)] or $\{\cL(a),\cL(c)\}$ is a contradicting cherry of $R$ and $\cF$ with $a\not\sim_{\cF}c$ and $e \in \{e_a,e_c\}$,
\item[(3)] or $\{\cL(a),\cL(c)\}$ is a contradicting cherry of $R$ and $\cF$ with $a\sim_{\cF}c$ and $e \in \{e_a,e_B,e_c\}$.
\end{enumerate}

In this context, $e_a$ and $e_c$ is an edge in $\cF$ adjacent to both leaves $a$ and $c$, respectively. Moreover, $e_B$ is part of the set of pendant edges for $\{\cL(a),\cL(c)\}$.\\

\textbf{Cherry lists.} Now, given two rooted binary phylogenetic $\cX$-trees $T_1$ and $T_2$, we say that $\fcL$ is a \emph{cherry list for $T_1$ and $T_2$}, if, while calling \textsc{processCherries}$(T_1,\{T_2\},\fcL)$, in the $i$-th iteration $\wedge_i$ is a cherry action for $R_i$ and $\cF_i$. Note that, if $\fcL$ is not a cherry list for $T_1$ and $T_2$, calling \textsc{processCherries}$(T_1,\{T_2\},\fcL)$ (cf.~Alg.~\ref{14-alg-proc}) returns the empty set.

\begin{algorithm}[]
\scriptsize
$M\leftarrow\emptyset$\;
\ForEach{$i \in 1,\dots,n$}{
 \If{\em{$\fc_i$ is a cherry action for $R$ and $\cF$}}{
 	$(\{\cL(a),\cL(c)\},e_i)\leftarrow \fc_i$\;
	\If{$e_i=\emptyset$}{
 		$M\gets \mbox{Add } \{\cL(a), \cL(c)\}\mbox{ as last element of } M$\;
		$R \gets R[\{\cL(a),\cL(c)\}\rightarrow \cL(a)\cup \cL(c)]$\;
		$\cF\gets\cF[\{\cL(a),\cL(c)\}\rightarrow \cL(a)\cup \cL(c) ]$\;
 	}
 	\Else{ 
 		$\cF\leftarrow\cF-\{e_i\}$\;
 	 	$R\leftarrow R|_{\cL(\overline\cF)}$\;
 	}
 }
 \Else{
 \Return{$\emptyset$}\;
 }
}
\While{$M$ {\em is not empty}}{ 
 $M\gets \mbox{remove the last element, say $\{\cL(a), \cL(c)\}$, from } M$\;
	$\cF\gets \cF[\cL(a)\cup \cL(c) \rightarrow \{\cL(a),\cL(c)\}]$\;
}
\Return{$R,\cF,M$}
\caption{\textsc{processCherries}$(R,\cF,\fcL=(\fc_1,\dots,\fc_n))$}
\label{14-alg-proc}
\end{algorithm}

Further details as well as an example of {\sc~processCherries} can be found in the work of Scornavacca \emph{et al.}~\citet{Scornavacca2012}.\\

Next, we will introduce the algorithm \scPCa. This algorithm is a simplified version of the algorithm \scAMa~describing one of its computational paths by a list of \emph{extended cherry actions}. Notice that this algorithm is a major tool that will help us to establish the correctness of our first modified algorithm \scAMa.\\

\textbf{Extended cherry action.} Let $R$ be a rooted binary phylogenetic $\cX$-tree and let $\cF$ be a forest for $R$. Then, an extended cherry action $\fc=(\{\cL(a),\cL(c)\},E')$ is a tuple containing a set $\{\cL(a),\cL(c)\}$ of two taxa of two leaves $a$ and $c$ as well as an edge set $E'$. We say that $\wedge$ is an extended cherry action for $R$ and $\cF$, if $\{\cL(a),\cL(c)\}$ is a contradicting cherry of $R$ and $\cF$ with $a\sim_{\cF}c$ and $E'=E_B$, where $E_B$ refers to the set of pendant edges for $\{\cL(a),\cL(c)\}$ in $\cF$.\\

\textbf{Extended cherry list.} Now, given two rooted binary phylogenetic $\cX$-trees $T_1$ and $T_2$, we say that $\fcL$ is an \emph{extended cherry list for $T_1$ and $T_2$}, if, while calling \scPCa$(T_1,\{T_2\},\fcL)$, in the $i$-th iteration $\wedge_i$ is either a cherry action or an extended cherry action for $R_i$ and $\cF_i$. Note that, if $\fcL$ is not an extended cherry list for $T_1$ and $T_2$, calling \scPCa$(T_1,\{T_2\},\fcL)$ (cf.~Alg.~\ref{14-alg-dagger_proc}) returns the empty set. \\

\begin{algorithm}[]
\scriptsize
$M\leftarrow\emptyset$\;
\ForEach{$i \in 1,\dots,n$}{
	\If{\em{$\fc_i$ is a cherry action for $R$ and $\cF$}}{
 	 	$(\{\cL(a),\cL(c)\},e_i)\leftarrow \fc_i$\;
 	 	\If{$e_i=\emptyset$}{
 	 		$M\gets \mbox{Add } \{\cL(a), \cL(c)\}\mbox{ as last element of } M$\;
			$R \gets R[\{\cL(a),\cL(c)\}\rightarrow \cL(a)\cup \cL(c)]$\;
			$\cF\gets\cF[\{\cL(a),\cL(c)\}\rightarrow \cL(a)\cup \cL(c) ]$\;
 		 }
 	 	\Else{ 
 	 		$\cF\leftarrow\cF-\{e_i\}$\;
 	 		$R\leftarrow R|_{\cL(\overline\cF)}$\;
 		 }
 	}
 	\ElseIf{\em{$\fc_i$ is an extended cherry action for $R$ and $\cF$}}{
 		$(\{\cL(a),\cL(c)\},E_i)\leftarrow \fc_i$\;
 		$\cF\leftarrow\cF-E_i$\;
 	 	$R\leftarrow R|_{\cL(\overline\cF)}$\;
 	}
 	\Else{
 	 \Return{$\emptyset$}\;
 	}
}
\While{$M$ {\em is not empty}}{ 
 $M\gets \mbox{remove the last element, say $\{\cL(a), \cL(c)\}$, from } M$\;
	$\cF\gets \cF[\cL(a)\cup \cL(c) \rightarrow \{\cL(a),\cL(c)\}]$\;
}
\Return{$R,\cF,M$}
\caption{\scPCa$(R,\cF,\fcL=(\fc_1,\dots,\fc_n))$}
\label{14-alg-dagger_proc}
\end{algorithm}

\begin{lemma}
Let $\fcL'$ be an extended cherry list for two rooted binary phylogenetic $\cX$-trees $T_1$ and $T_2$. Moreover, let $\cF$ be a forest for $T_2$ calculated by $\fcL'$. Then, there also exists a cherry list $\fcL$ for $T_1$ and $T_2$ calculating $\cF$.
\label{14-lem-ext}
\end{lemma}

\begin{proof}
To proof this lemma we will show how to replace each extended cherry action of $\fcL'$ so that the resulting cherry list $\fcL$ still computes $\cF$. Let $\fc_i'=(\{\cL(a),\cL(c)\},E_B)$ be an extended cherry action with $E_B=\{e_1,e_2,\dots,e_k\}$. Then, we can replace $\fc_i'$ through the sequence of cherry actions $$(\{\cL(a),\cL(c)\},e_1),(\{\cL(a),\cL(c)\},e_2),\dots,(\{\cL(a),\cL(c)\},e_k).$$ Since by these cherry actions the same edges are cut from $T_2$ as by $\fc_i'$, the topology of $R_{i+1}'$ equals $R_{i+k+1}$, which directly implies that $\cF$ is still computed.
\end{proof}

\subsection{Correctness of \scAMa}
\label{sec-proof3}

In this section, we will discuss the correctness of our first modified algorithm \scAMa~by establishing the following theorem.

\begin{theorem}
Let $T_1$ and $T_2$ be two rooted binary phylogenetic $\cX$-trees and $k\in\mathbb{N}$. Calling $$\text{\scAMa$(T_1,T_2,T_1,\{T_2\},k,\emptyset)$}$$ returns all maximum acyclic agreement forests for $T_1$ and $T_2$ if and only if $k\ge h(T_1,T_2).$
\label{14-th-dagger}
\end{theorem}

\begin{proof}

Let $T_1$ and $T_2$ be two rooted binary phylogenetic $\cX$-trees and let $\fcL$ be a cherry list for $T_1$ and $T_2$ mimicking a computational path of \textsc{allMAAFs} calculating a maximum acyclic agreement forest $\cF$ for $T_1$ and $T_2$. Then, in the following, we say a cherry action $\fc_i=(\{\cL(a),\cL(c)\},e_i)$ in $\fcL$ is a \emph{special cherry action}, if $\{\cL(a),\cL(c)\}$ is a contradicting cherry of $R_i$ and $\cF_i$ and if $e_i$ is contained in the set of pendant edges for $\{\cL(a),\cL(c)\}$. Note that, whereas such special cherry actions may occur in a computational path of \textsc{allMAAFs}, this is not the case for a computational path of \scAMa. In the following, however, we will show that the algorithm \textsc{allMAAFs} calculates a maximum acyclic agreement forest $\cF$ for $T_1$ and $T_2$ if and only if $\cF$ is calculated by \scAMa.\\

\begin{lemma}
Let $T_1$ and $T_2$ be two rooted binary phylogenetic $\cX$-trees and let $\cF$ be a maximum acyclic agreement forest for $T_1$ and $T_2$ of size $k$. Then, $\cF$ is calculated by calling $\textsc{allMAAFs}(T_1,T_2,T_1,T_2,k,\emptyset)$ if and only if $\cF$ is calculated by calling \scAMa$(T_1,T_2,T_1,\{T_2\},k,\emptyset).$
\label{14-lem-dagger}
\end{lemma}

\begin{proof}

'$\Longrightarrow$': From Lemma~\ref{14-lem-ext} we can directly deduce that if there exists a computational path of \scAMa~calculating $\cF$, then, there also exists a computational path in \textsc{allMAAFs} calculating $\cF$.\\

'$\Longleftarrow$': Let $\fc_j=(\{\cL(a),\cL(c)\},e_j=\emptyset)$ be a cherry action of $\fcL$ contracting both leaves $a$ and $c$. Then, we say a preceding special cherry action $\fc_i=(\{\cL(x),\cL(y)\},e_i)$ ($i<j$) refers to $\fc_k$ if the following condition is satisfied. Let $a_i$ and $c_i$ be the lowest common ancestor of $\cL(a)$ and $\cL(c)$ in $R_i$, respectively. Then, $e_i\neq\emptyset$ has to be a pendant edge lying on the path connecting both nodes $a_i$ and $c_i$. In such a case, we say the special cherry action is either of \emph{Type~A} or \emph{Type~B} (see definition below). Otherwise, if a special cherry action does not refer to another cherry action, we say this cherry action is of \emph{Type~C}.

Now, based on the edge $e_i=(v,w)$, being part of the special action $\fc_i$ referring to the cherry action $\fc_j$, we further distinguish whether $\fc_i$ is either of \emph{Type~A} or of \emph{Type~B}. Therefore, let $T(w)$ be the subtree rooted at $w$ and let $\cL(T(w))$ be the set of taxa being contained in $T(w)$. Now, we say $\fc_i$ is of \emph{Type~A} if there exists a forest $\cF_k$ and a tree $R_k$ with $i<k<j$ so that both of the following two conditions are satisfied.

\begin{itemize}
\item[(i)] Each taxon in $\cL(T(w))$ is part of a taxa set of an isolated node. Notice that, as a direct consequence, there exists a leaf $w'$ in $R_k$ with label $\cL(w)$.
\item[(ii)] The sibling $s'$ of $w'$ in $R_k$ is a leaf. Notice that, as a direct consequence, $\{\cL(s'),\cL(w')\}$ is a cherry of $R_k$.
\end{itemize}

Otherwise, if there does not exist such a forest $\cF_k$ and such a tree $R_k$ satisfying these two conditions, we say $\fc_i$ is of \emph{Type~B}. An illustration of these two types is given in Figure~\ref{14-fig-typeABC}

Notice that, due to the so chosen definition, each special cherry action in $\fcL$ has to be either of \emph{Type~A}, of \emph{Type~B}, or of \emph{Type~C}. Next, based on this observation, we will show in three steps how to turn $\fcL$ into another cherry list for $T_1$ and $T_2$ not containing any special cherry actions, but still calculating $\cF$, which can be briefly summarized as follows.

In a first step, we will show how to modify $\fcL$ by replacing each special cherry action of \emph{Type~A} and of \emph{Type~C} through a non-special cherry action so that the result is still a cherry list for $T_1$ and $T_2$ calculating $\cF$. Next, during a second and a third step, we will show how to replace each set of special cherry actions of \emph{Type~B} all referring to the same cherry action by a single extended cherry action so that the result is an extended cherry list for $T_1$ and $T_2$ still calculating $\cF$.\\

\begin{figure}
\centering
\includegraphics[scale=1.3]{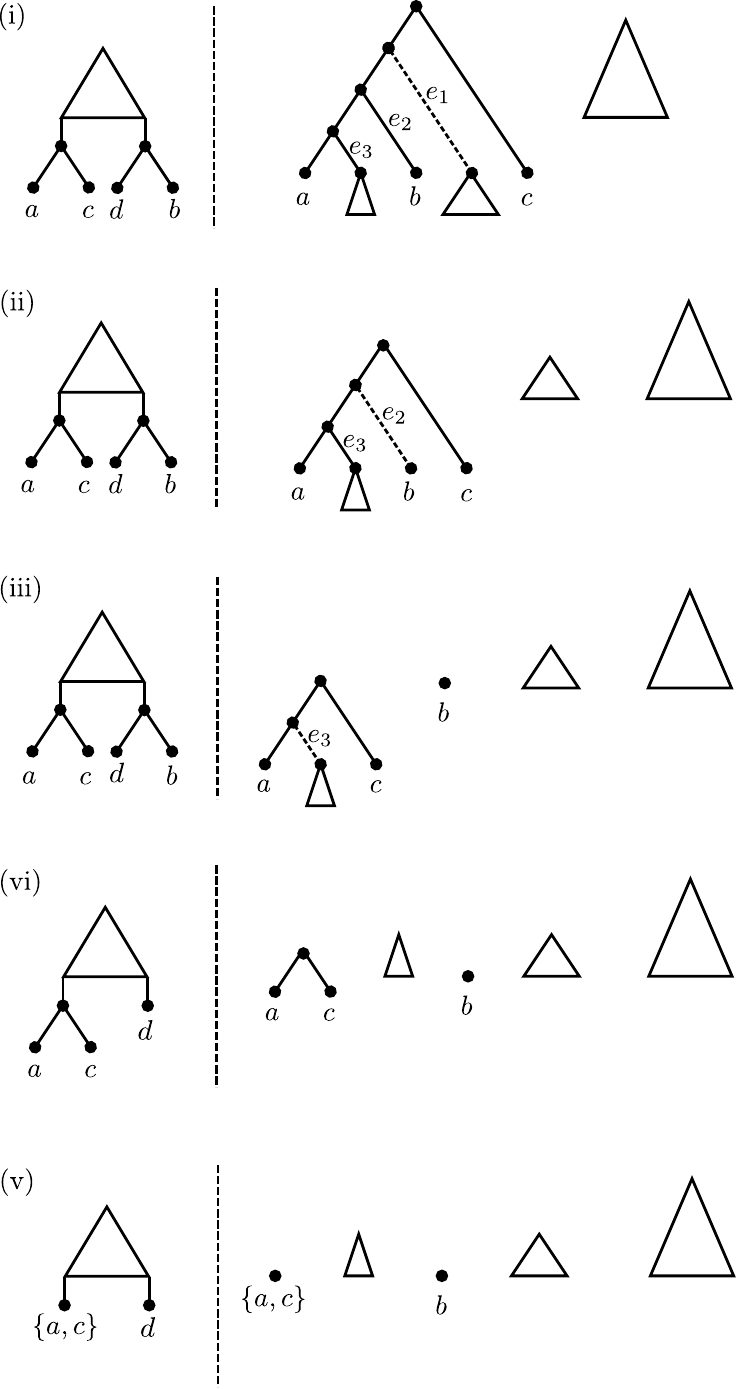}
\caption{An illustration of the cherry list $((\{a,c\},e_1),(\{a,c\}),e_2),(\{a,c\}),e_3),(\{a,c\}),\emptyset))$ in which $(\{a,c\},e_1)$ is a special cherry action of \emph{Type~B} and  $(\{a,c\},e_2)$ is a special cherry action of \emph{Type~A}. Note that the expanded cherry list $((\{b,d\},e_2),(\{a,c\},\{e_1,e_3\}),(\{a,c\}),\emptyset))$ applied to (i) yields the same scenario as depicted in (v) without making use of any special cherry actions.}
\label{14-fig-typeABC}
\end{figure} 

\textbf{Step 1.} Let $\fc_i=(\{\cL(x),\cL(y)\},e_i=(v,w))$ be a special cherry action of \emph{Type~A} or of \emph{Type~C}, and let $\fc_k$ with $i<k<j$ be the first cherry action of $\fcL$ in which in $\cF_k$ each taxon in $\cL(T(w))$ is contained in a taxa set of an isolated node so that in $R_k$ the sibling $s'$ of the leaf labeled by $\cL(w)$ is also leaf. Then, let $\fcL'$ be a cherry list that is obtained from $\fcL$ by first removing $\fc_i=(\{\cL(x),\cL(y)\},e_i)$ and then by inserting the cherry action $(\{\cL(w'),\cL(s')\},e_{w'})$ right after $\fc_k$, so that $\fcL'$ equals

$$(\wedge_1',\dots,\wedge_n')=(\wedge_1,\dots,\wedge_{i-1},\wedge_{i+1},\dots,\wedge_{k},\wedge_{k}',\dots,\wedge_{n}'),$$

\noindent where $\wedge_{k}'=(\{\cL(w'),\cL(s')\},e_{w'})$ with $e_{w'}$ being the in-edge of $w'$.

Now, for the following, remember that $R_l$ (or $R_l'$) refers to the input tree occurring during iteration $l$ while processing the cherry list $\fcL$ (or $\fcL'$). Moreover, we write $R_i=_{\wedge}R_j'$ if both trees contain the same set of cherries. Then, based on the position of the cherry action $\fc_l$ in $\fcL$, we can make the following five observations (cf.~Fig.~\ref{14-fig-dagger-1}).

\begin{figure}[tb]
\centering
\includegraphics[scale=0.7]{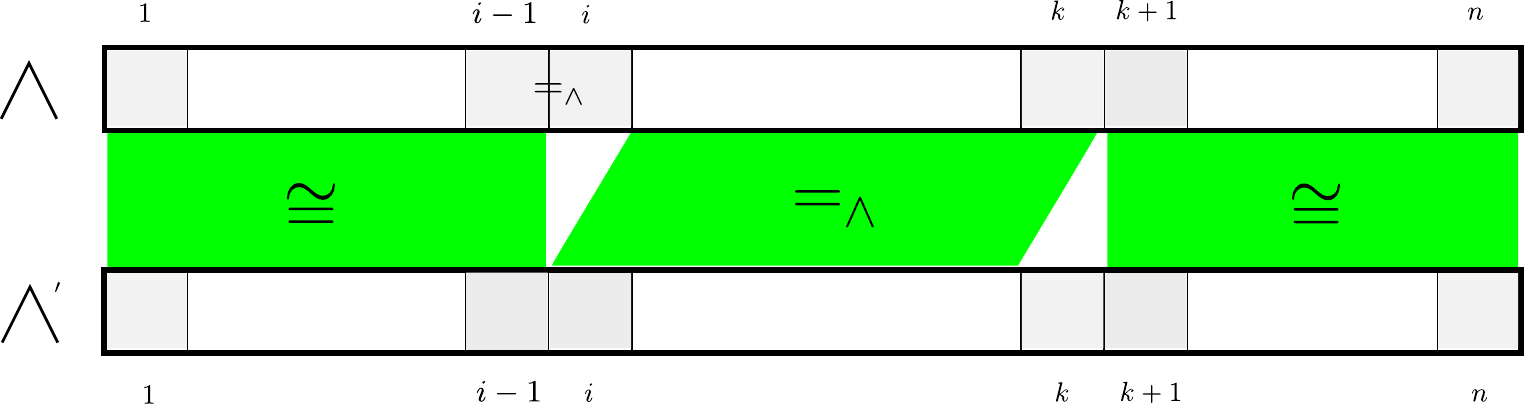}
\caption{An illustration of the scenario as described in Step 1. The figure shows a comparison between two sequences of $R_i$-trees; one corresponding to the original cherry list and one corresponding to the modified cherry list, in which a special cherry action of \emph{Type A} is shifted from position $i$ to $k+1$.}
\label{14-fig-dagger-1}
\end{figure} 

\begin{itemize}
\item If $l<i$, then $R_{l}'$ equals $R_l$: This is the case because $R_1$ equals $R_1'$ and through $\fc_l$ and $\fc_l'$ the same tree operation is performed on $R_l$ and $R_l'$, respectively.
\item If $l=i$, then $R_{l-1}'=_{\wedge}R_l$: This is the case because $R_{i-1}$ equals $R_{i-1}'$, in $R_l$ the node $w$ cannot be part of a cherry (due to the definition of a special cherry action of \emph{Type~A}) and, thus, through $\fc_l$ (=$\fc_i$) the set of cherries in $R_{i-1}$ remains unchanged.
\item If $i<l<k+1$, then $R_{l-1}'=_{\wedge}R_l$: This is the case because $R_i=_{\wedge}R_{i-1}'$ and again in $R_l$ node $w$ cannot be part of a cherry. Thus, the two cherry actions $\fc_l$ and $\fc_{l-1}'$ produce the same set of cherries in $R_l$ and $R_{l-1}'$.
\item If $l=k+1$, then $R_{k+1}'$ equals $R_{k+1}$: This is the case because $R_k=_{\wedge}R_{k-1}'$ and through $\fc_k'$ first node $w$ is cut and then removed from $R_{k-1}'$. Consequently, through the cherry actions $\fc_q$ and $\fc_q'$, with $1\le q<k+1$, the the same tree operations are applied to $R_1$ and $R_1'$ and, thus, $R_{k+1}$ equals $R_{k+1}'$.
\item If $l>k+1$, then $R_{l}'$ equals $R_l$: This is the case because $R_{k+1}$ equals $R_{k+1}'$ and through $\fc_l$ and $\fc_l'$ the same tree operation is performed on $R_l$ and $R_l'$, respectively.
\end{itemize}

Due to these observations and since through $\fcL'$ still the same edges are cut from $T_2$ as through the original cherry list $\fcL$, $\fcL'$ is a cherry list for $T_1$ and $T_2$ still calculating $\cF$. Now, by consecutively replacing all special cherry actions of \emph{Type~A} and \emph{Type~C} we can turn $\fcL$ into the cherry list $\fcL^{(1)}$ for $T_1$ and $T_2$ only containing special cherry actions of \emph{Type~B} and still computing $\cF$. Next, we will show how to remove each of those remaining special cherry actions.\\

\textbf{Step 2.} Let $\fcL^*\subset\,\fcL$ with $|\fcL^*|=k$ be a set of special cherry actions of \emph{Type~B} all referring to a cherry action $\fc_j=(\{\cL(a),\cL(c)\},\emptyset)$, and let $i=\min_{i'}\{\fc_{i'}:\fc_{i'}\in\,\fcL^*\}$  with $\fc_i=(\{\cL(x),\cL(z)\},e_{i})$. Moreover, let $\fcL'$ be the cherry list that is obtained from $\fcL^{(1)}$ as follows. First the cherry of each cherry action in $\fcL^*$ is replaced by $\{\cL(a),\cL(c)\}$ and then all those cherry actions in $\fcL^*$ are rearranged such that they are placed in sequential order directly right before position $j$. 

This means, in particular, that $\fcL'$ contains a sequence of special cherry actions

\begin{flalign*}
&~\vdots\\
\wedge_{j-k} &= (\{\cL(a),\cL(c)\},e_{{j-k}}=(v_{j-k},w_{j-k})),\\
\wedge_{j-k+1} &= (\{\cL(a),\cL(c)\},e_{j-k+1}=(v_{j-k+1},w_{j-k+1})),\\
&~\vdots\\
\wedge_{j-1} &= (\{\cL(a),\cL(c)\},e_{j-1}=(v_{j-1},w_{j-1})),\\
\wedge_{j} &= (\{\cL(a),\cL(c)\},\emptyset),\\
&~\vdots
\end{flalign*}

\noindent in which each edge of $\fcL^*$ is contained in the set of pendant edges for the cherry $\{\cL(a),\cL(c)\}$. Now, just for convenience, in the following we assume, without loss of generality, that all those cherry actions in $\fcL^*$ occur in sequential order beginning at position $i$. Moreover, for the following, let $\cW$ be the set of target nodes of each edge in $\fcL^*$ and let $R_l$ (or $R_l'$) be the input tree of iteration $l$ while processing the cherry list $\fcL^{(1)}$ (or $\fcL'$). Additionally, again we write $R_i=_{\wedge}R_j'$ if both trees contain the same set of cherries. Then, based on the position of a cherry action $\fc_l$ in $\fcL$, we can make the following five observations (cf.~Fig.~\ref{14-fig-dagger-2}).

\begin{figure}[tb]
\centering
\includegraphics[scale=0.7]{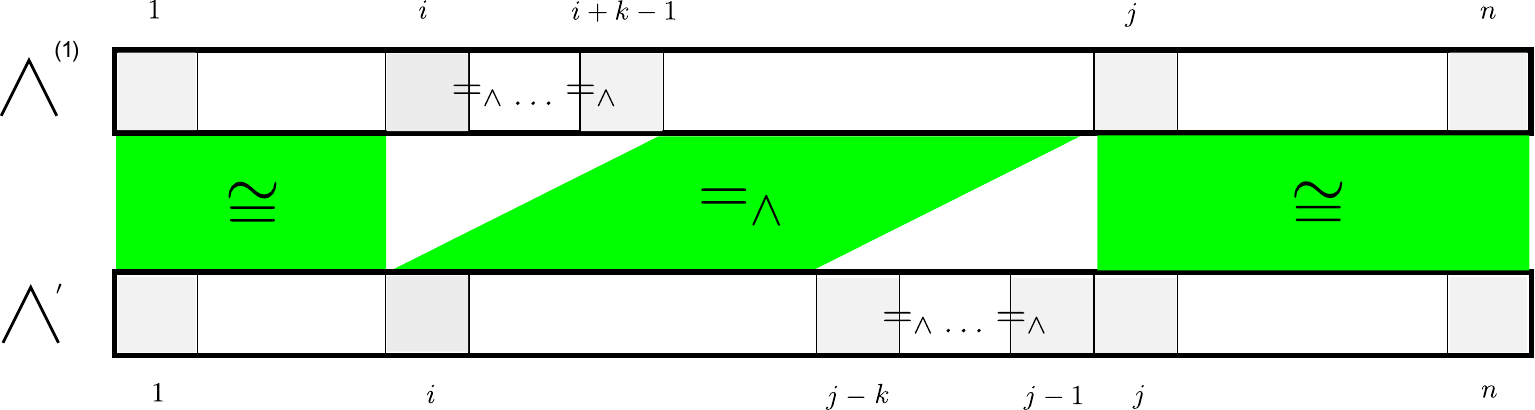}
\caption{An illustration of the scenario as described in Step 2. The figure shows a comparison of the two sequences of $R_i$-trees; one corresponding to the original cherry list and one corresponding to the modified cherry list, in which a sequence of $k$ special cherry actions of \emph{Type B}, beginning at position $i$, is shifted to position $j-k$.}
\label{14-fig-dagger-2}
\end{figure} 

\begin{itemize}
\item If $l<i$, then $R_l'$ equals $R_l$: This is the case because $R_1$ equals $R_1'$ and through $\fc_l$ and $\fc_l'$ the same tree operation is performed on $R_l$ and $R_l'$, respectively.
\item If $i-1<l<i+k$, then $R_l=_{\wedge}R'_{i-1}$: This is the case because $R_{i-1}$ equals $R_{i-1}'$ and, since in $R_l$ each $w_i$ in $\cW$ cannot be part of a cherry, through $\fc_l$ the set of cherries in $R_l$ remains unchanged.
\item If $i+k-1<l<j$, then $R_{l-k}'=_{\wedge}R_{l}$: This is the case because $R_{i+k-1}=_{\wedge}R_{i-1}'$ and again in $R_l$ each $w_i$ in $\cW$ cannot be part of a cherry. Thus, the two cherry actions $\fc_l$ and $\fc_{l-k}'$ produce the same set of cherries in $R_l$ and $R_{l-k}'$.
\item If $l=j$, then $R_{l}'$ equals $R_{l}$: This is the case because $R_{j-1}=_{\wedge}R_{j-k-1}'$ and through each cherry action in $\fcL^*$ each subtree $T(w_i)$ is cut (and removed from $R_{j-k-1}'$, if $w_i$ is a leaf). Consequently, through the cherry actions $\fc_q$ and $\fc_q'$, with $1\le q<j$, the same tree operations are applied to $R_1$ and $R_1'$ and, thus, $R_{j}$ equals $R_{j}'$.
\item If $l>j$, then $R_{l}'$ equals $R_l$: This is the case because $R_{j}$ equals $R_{j}'$ and through $\fc_l$ and $\fc_l'$ the same tree operation is performed on $R_l$ and $R_l'$, respectively.
\end{itemize}

Now, by consecutively rearranging all special cherry actions of \emph{Type~B} as described above, we can turn $\fcL$ into the cherry list $\fcL^{(2)}$ for $T_1$ and $T_2$ in which all special cherry actions referring to the same cherry action are located next to each other. Moreover, as a direct consequence of these observations and since through $\fcL^{(2)}$ still the same edges are cut from $T_2$ as through $\fcL^{(1)}$, $\fcL^{(2)}$ is still a cherry list for $T_1$ and $T_2$ calculating~$\cF$.\\

\textbf{Step 3.} Let $\fcL^{(2)}$ be the list that is obtained from applying Step 1 and Step 2 as described above. Then, we can further modify the cherry list $\fcL^{(2)}$ to $\fcL^{(3)}$ by replacing each sequence 

$$(\wedge_i=(\{\cL(a),\cL(c)\},e_i),\dots,\wedge_{i+k-1}=(\{\cL(a),\cL(c)\},e_{i+k-1}))$$ 

\noindent of special cherry actions (i.e., cherry actions of \emph{Type~B}) through a single extended cherry action $\fc_B=(\{\cL(a),\cL(c)\},E_B)$, where $E_B$ denotes the pendant edge set for $\{\cL(a),\cL(c)\}$. As a consequence, since $\fc_B$ just summarizes all tree operations conducted by the replaced sequence of special cherry actions, by $\fcL^{(3)}$ the same edges are cut from $T_2$ as it is the case for the cherry list $\fcL^{(2)}$. As a direct consequence, the maximum acyclic agreement forest $\cF$ is still calculated by the extended cherry list $\fcL^{(3)}$.\\

In summary, as described by those three steps, we can consecutively replace all special cherry actions of $\fcL$ so that the resulting list $\fcL^{(3)}$ satisfies all of the following three conditions.

\begin{itemize}
\item $\fcL^{(3)}$ is an extended cherry list for $T_1$ and $T_2$.
\item $\fcL^{(3)}$ does not contain any special cherry actions.
\item $\fcL^{(3)}$ calculates the maximum acyclic agreement forest $\cF$.
\end{itemize}

This directly implies that for each computational path in \textsc{allMAAFs} calculating a maximum acyclic agreement forest $\cF$ for both input trees there also exists a computational path in \scAMa~calculating $\cF$.
\end{proof}

Now, from Lemma~\ref{14-lem-dagger} we can directly deduce the correctness of Theorem~\ref{14-th-dagger}.
\end{proof}

\subsection{Correctness of \scAMc}
\label{sec-proof2}

The correctness of our third modified algorithm \scAMc~principally directly follows from the correctness of the refinement step, which can be found in the work of Whidden \emph{et al.}~\citet{Whidden2011}. We still have to show, however, that omitting both additional cutting steps when processing a common cherry still computes all of those agreement forests from which all maximum acyclic agreement forest can be obtained by cutting some of its edges.

\begin{theorem}
Let $T_1$ and $T_2$ be two rooted binary phylogenetic $\cX$-trees and $k\in\mathbb{N}$. Calling $$\text{\scAMc$(T_1,T_2,T_1,\{T_2\},k,\emptyset)$}$$ returns all maximum acyclic agreement forests for $T_1$ and $T_2$ if and only if $k\ge h(T_1,T_2).$ 
\label{14-lem-three}
\end{theorem}

\begin{proof}
To proof Theorem~\ref{14-lem-three} we first have to establish a further lemma. Here, we argument that by contracting common cherries instead of cutting one of its edges, results in an agreement forests containing less components.

\begin{lemma}
Let $\fcL=(\fc_1,\fc_2,\dots,\fc_n)$ be a cherry list for two rooted binary phylogenetic $\cX$-trees $T_1$ and $T_2$ with $\fc_i=(\{\cL(a),\cL(c)\},e_c)\in\fcL$ so that $\{\cL(a),\cL(c)\}$ is a common cherry of $R_i$ and $\cF_i$ and $e_c$ the in-edge of leaf $c$. Moreover, let $\cF$ be the agreement forest for $T_1$ and $T_2$ that is calculated by calling {\sc~processCherries}$(R_1,\cF_1,\fcL)$ such that $|\cF|=k$. Then, there exists an agreement forest $\hat\cF$ with $|\hat\cF|=k-1$ that can be computed by calling {\sc~processCherries}$(R_i,\cF_i,(\fc_i',\dots,\fc_n'))$ in which $\fc_i'=(\{\cL(a),\cL(c)\},\emptyset)$.
\label{14-lem-one}
\end{lemma}

\begin{proof} 
Let $F_a$ and $F_c$ be the two components in $\cF$ that have been derived from expanding both components containing $a$ and $c$, respectively, and let $\cX_a$ be the taxa set of the subtree rooted at $a$ (i.e., $F_a(a)$). Then, by attaching $F_c$ back to $F_a$ we can reduce the size of the agreement forest $\cF$ by one. Here, depending on whether $\scLCA_{F_a}(\cX_a)$ is the root of $F_a$ or not, this re-attachment step can be done in two different ways. 

\begin{itemize}
\item If $\scLCA_{F_a}(\cX_a)$ is not the root of $F_a$ and, thus, has an in-going edge $e=(u,w)$, first $e$ is split into two adjacent edges $(u,v)$ and $(v,w)$ and then $F_c$ is re-attached to $v$ by inserting a new edge $(v,r_c)$, where $r_c$ denotes the root of $F_c$. 
\item Otherwise, if $\scLCA_{F_a}(\cX_a)$ is the root of $F_a$, first a new node $v$ is created and then $v$ is connected to the two roots of $F_a$ and $F_c$. 
\end{itemize}

In the following, we will denote the component that is obtained from attaching $F_c$ back to $F_a$ by $F_r$ and the resulting set of components by $\hat\cF$. In Figure~\ref{14-fig-2forests}, we give an illustration of the two forests $\cF$ and $\hat\cF$.

\begin{figure}
\centering
\includegraphics[scale=1]{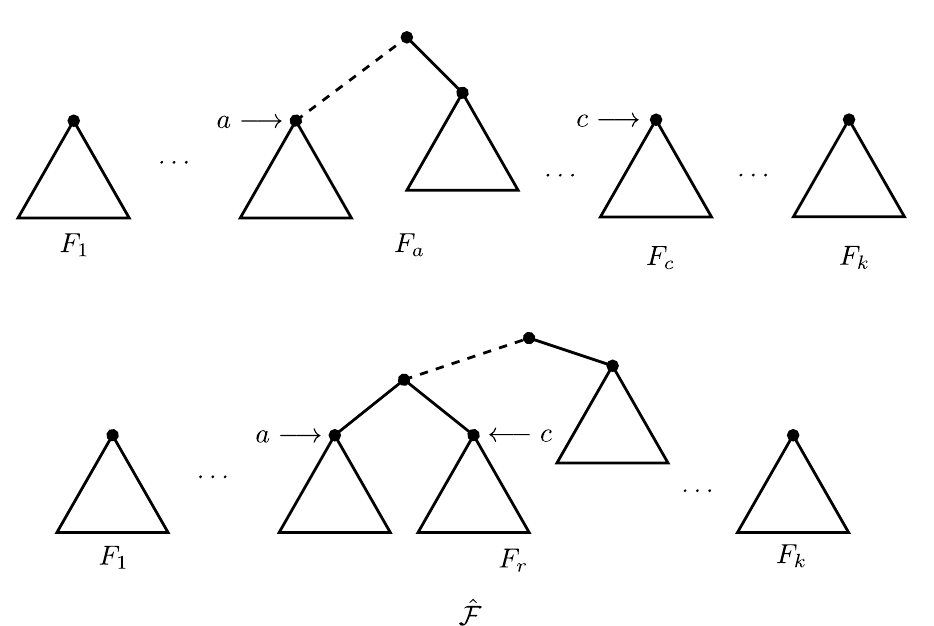}
\caption[Two forests as defined in Lemma~\ref{14-lem-one}]{The two forests $\cF$ and $\hat\cF$ as defined in Lemma~\ref{14-lem-one}. Note that, since $F_c$ is attached to $F_a$, the size of $\hat\cF$ is $|\cF|-1$.} 
\label{14-fig-2forests}
\end{figure} 

Notice that we can calculate the agreement forest $\hat\cF$ by calling the algorithm{\sc~processCherries} with a specific cherry list $\fcL'$ that can be derived from the original cherry list $\fcL=(\fc_1,\fc_2,\dots,\fc_i=(\{\cL(a),\cL(c)\},e_c),\dots,\fc_n)$ as follows.

\begin{enumerate}
\item[(1)] Replace $\fc_i$ by $(\{\cL(a),\cL(c)\},\emptyset)$.
\item[(2)] In each subsequent cherry action $\fc_{i+1},\dots,\fc_n$, replace $\cL(a)$ through $\cL(a)\cup\cL(c)$.
\end{enumerate}

Let $R_{i+1}$ and $\cF_{i+1}$ be computed by applying $(\{\cL(a),\cL(c)\},e_c)$ to $R_i$ and $\cF_i$ and let $R'_{i+1}$ and $\cF'_{i+1}$ be computed by applying $(\{\cL(a),\cL(c)\},\emptyset)$ to $R_i$ and $\cF_i$. Since $\{\cL(a),\cL(c)\}$ is a common cherry of $R_i$ and $\cF_i$, the only difference between the two topologies, disregarding node labels, is that $\cF_{i+1}$ contains an additional component consisting of an isolated node labeled by $\cL(c)$. Notice, however, that, since this component is fully contracted, this component cannot have an impact on the subsequent cherry actions $\fc_{i+1},\dots,\fc_n$. Consequently, the cherry list $\fcL'$ calculates $\hat\cF$. 

Up to now, we have shown that there exists a set of components $\hat\cF$ of size $|\cF|-1$, which can be computed by calling the algorithm {\sc processCherries} with a slightly modified cherry list $\fcL'$. Now, to establish Lemma~\ref{14-lem-one}, we still have to show that $\hat\cF$ is an agreement forest for both input trees $T_1$ and $T_2$. First, note that each expanded forest to which \textsc{allMAAFs} applies the acyclic check (cf.~Alg.~\ref{14-alg-main},~line~5) has to satisfy each condition of an agreement forest \cite[Lemma 3]{Scornavacca2012}. As $\fcL'$ is imitating a computational path corresponding to the original algorithm \textsc{allMAAFs} calculating $\hat F$, this directly implies that $\hat F$ is an agreement forest for $T_1$ and $T_2$. 
\end{proof}

Now, given two rooted binary phylogenetic $\cX$-trees $T_1$ and $T_2$, assume that our first modified algorithm \scAMa~contains a computational path calculating a maximum acyclic agreement forest $\cF$ of size $k$ by cutting instead of contracting a common cherry. More precisely, let $\fcL$ with $\fc_i=(\{\cL(a),\cL(c)\}\in\fcL$ be the cherry list mimicking this computational path and, without loss of generality, let $\fc_i$ be the only cherry action cutting instead of contracting a common cherry. Then, as already discussed in the proof of Lemma~\ref{14-lem-one}, there additionally exists a cherry list $\fcL'$ calculating a specific agreement forest $\hat\cF$ of size $k-1$ containing a component $\hat F$. More specifically, let $F_a$ and $F_c$ be the expanded component in $\cF$ which is derived from the contracted node $a$ and $c$, respectively. Then, $\hat\cF$ is obtained from $\cF$ by re-attaching $F_c$ back to $F_a$ (cf.~Fig.~\ref{14-fig-2forests}). Thus, as a direct consequence, by cutting the in-edge corresponding to the root of $\hat F(\cL(F_c))$ the maximum acyclic agreement forest $\cF$ arises. 

This implies that, if cutting instead of contracting a common cherry yields a maximum acyclic agreement forest $\cF$, the algorithm \scAMc~guarantees the computation of an agreement forest $\hat\cF$ from which $\cF$ can be obtained by cutting some of its edges. Note that the refinement step is always able to identify the minimal number of such edges and, thus, as the only difference between the first modified algorithm and the algorithm \scAMc~consists in the way a common cherry is processed, Theorem~\ref{14-lem-three} is established.
\end{proof}

\section{Simulation Study}
\label{14-sec-simulation}

Our simulation study has been conducted on the same synthetic dataset as the one used for the simulation study reported in the work of Albrecht \emph{et al.}~\citet{Albrecht2011}. It consists of binary phylogenetic tree pairs that have been generated in respect to one combination of the following three parameters: the \emph{number of taxa} $n=\{20,50,100.200\}$, the \emph{executed number of rSPR-moves} $k=\{5,10,\dots,50\}$, and the \emph{tangling degree} $d=\{3,5,10,15,20\}$. For each combination of those parameters $10$ tree pairs have been generated, resulting in $2000$ tree pairs in total. 

More precisely, a tree pair $(T_1,T_2)$ is computed as follows. In a first step, the first tree $T_1$ with $n$ leaves is computed which is done initially by randomly selecting two nodes $u$ and $v$ of a specific set $\cV$ consisting of $n$ nodes of both in- and out-degree $0$. Then, those two selected nodes $u$ and $v$ are connected to a new node $w$ and, finally, $\cV$ is updated by replacing $u$ and $v$ by its parent $w$. This process is repeated until $\cV$ consists only of one node corresponding to the root of $T_1$. In a subsequent step, $T_2$ is computed by applying $k$ rSPR-moves to $T_1$ each respecting tangling degree $d$.

When performing a sequence of rSPR-moves, one can undo or redo some of those moves and, thus, $k$ is only an upper bound of the real underlying rSPR-distance corresponding to both trees of a tree pair. The tangling degree, as already described in the work of Albrecht \emph{et al.}~\citet{Albrecht2011}, is an \emph{ad hoc} concept controlling the number of minimum common clusters during the construction of a tree pair. Since, however, all the four implementations that are tested on the synthetic dataset perform a cluster reduction and, thus, equally benefit from the number of minimum common clusters, this number is irrelevant for our simulation study and, consequently, we do not give any further details about this parameter. Instead, we refer the interested reader to the work of Albrecht \emph{et al.}~\citet{Albrecht2011}.

In this section, the practical runtimes produced by the respective implementations of our two modifications are compared with the practical runtimes of an implementation of the original algorithm. More precisely, by applying our synthetic dataset the practical runtime was measured for an implementation of the original algorithm \textsc{allMAAFs}, the algorithm \scAMa, and the algorithm \scAMc. Each of those algorithms has been integrated as a plug-in into the freely available Java software Hybroscale\footnote{\url{www.bio.ifi.lmu.de/softwareservices/hybroscale}}. The simulation study has been conducted on a grid computer providing $16$ cores and $40$\,GB RAM. In order to receive a reasonable set of completed data sets within an appropriate time period, we decided to compute only the hybridization number and omitted the computation of all maximum acyclic agreement forests. Moreover, we set the maximum runtime of each tree pair to $20$~minutes, which means that each tree pair whose computation of the hybridization number could not be finished within $20$~minutes was aborted. Notice that depending on the runtime analysis, which is shown in the following, those aborted tree pairs were either not taken into account or counted as being finished after $20$~minutes.

The problem arising when computing all, instead of just one, maximum acyclic agreement forests for certain tree sets is that, typically, there exists a large number of those agreement forests all being distributed in a vast search space. Consequently, even if those tree sets are of low computational complexity, one has to investigate far more than $20$~minutes, which means that one could not conduct the simulation study within an appropriate time period. Otherwise, if we would choose a relatively small maximum running time, such as $20$~minutes, we could only process those tree sets of low computational complexity, which are not able to indicate the speedup obtained from applying our three modifications.

\begin{figure}
\centering
\includegraphics[scale=0.4]{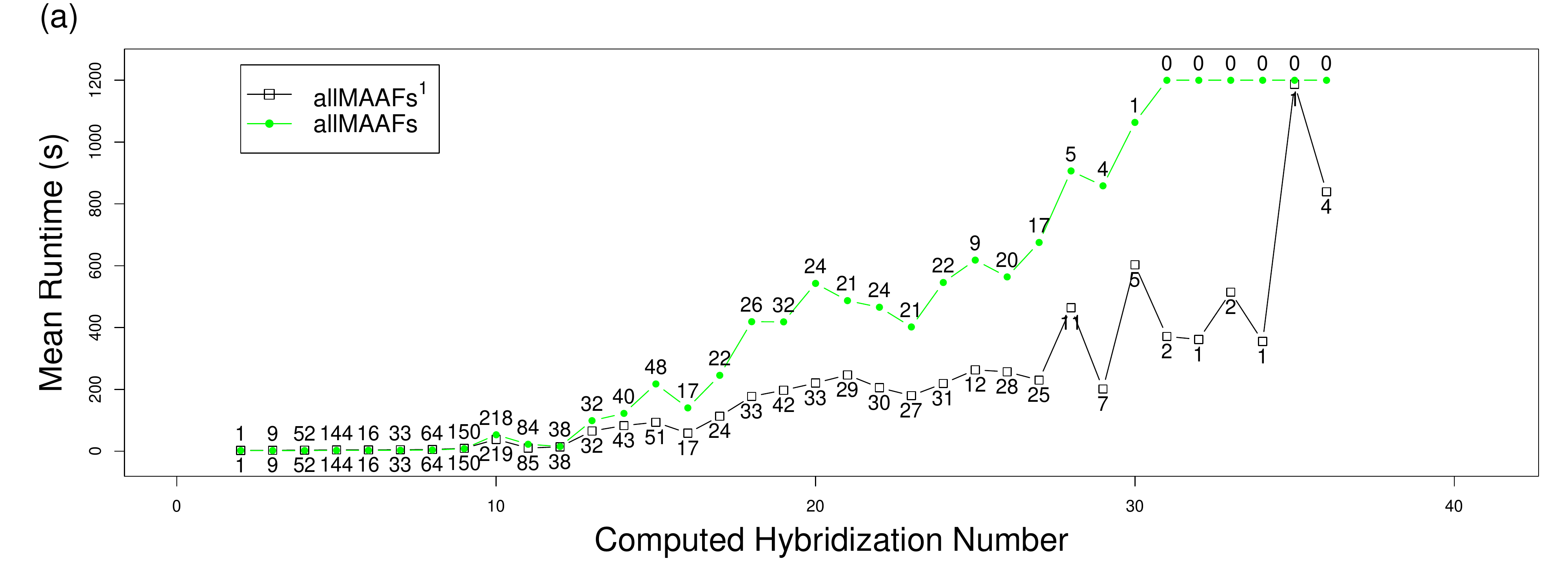}
\includegraphics[scale=0.4]{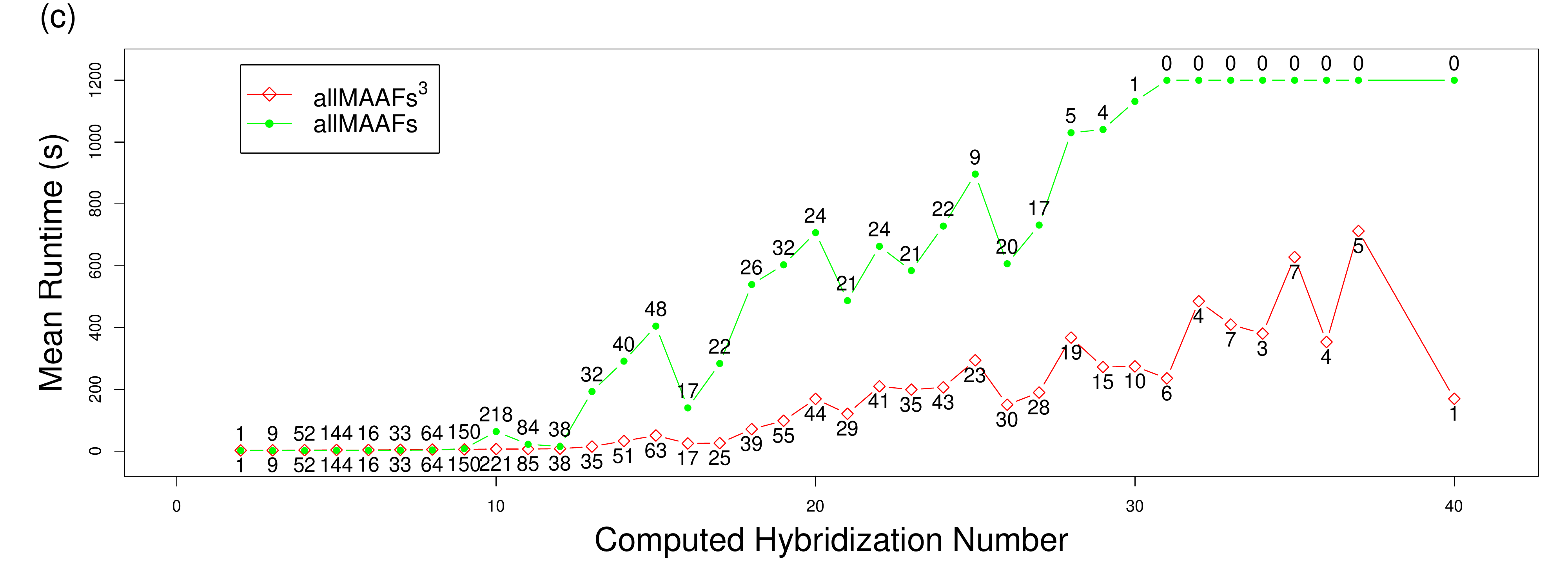}
\includegraphics[scale=0.4]{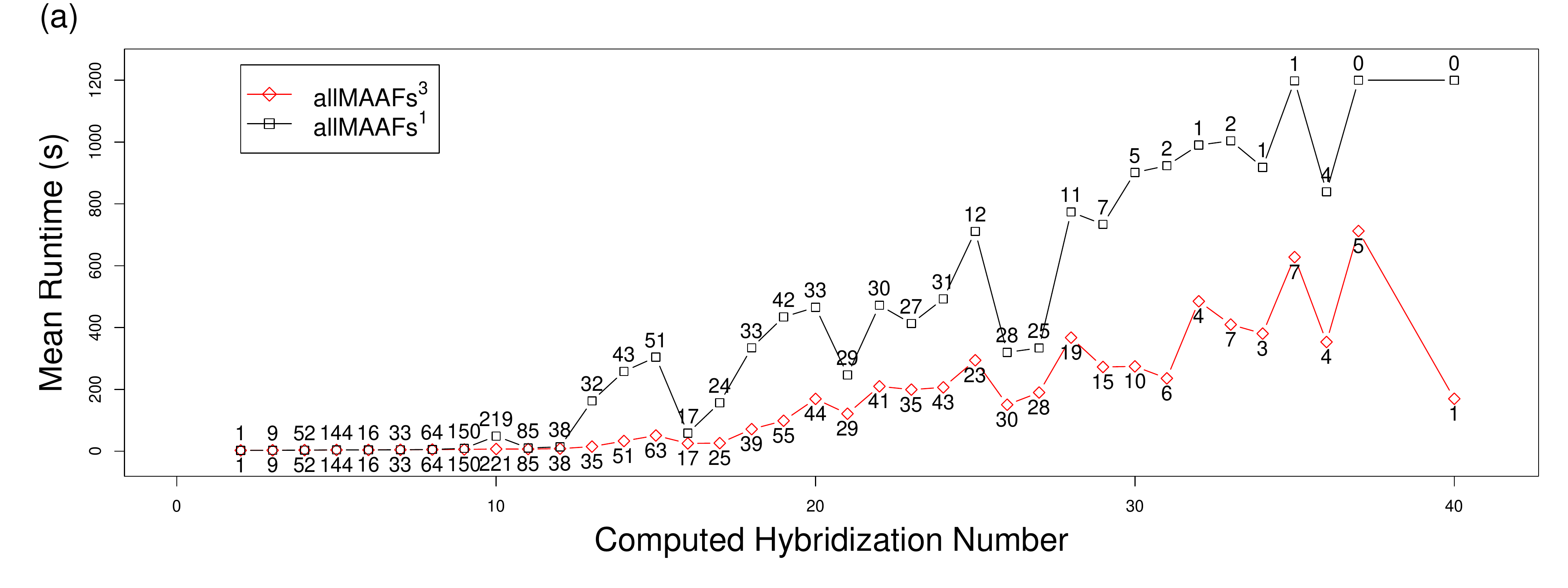}
\caption{Comparisons between the mean average runtimes in terms of the hybridization number of (a) \textsc{allMAAFs} and \scAMa,(b) \textsc{allMAAFs} and \scAMc, as well as (c) \scAMa~and \scAMc. For each hybridization number the corresponding number of tree pairs is given that could be computed within the time limit of $20$~minutes by the two respective algorithms and, thus, contributed to the mean average runtime denoted by the y-axis.} 
\label{14-fig-Hsim1}
\end{figure}

In Figure~\ref{14-fig-Hsim1}, the mean average runtimes in terms of the computed hybridization numbers are shown. More precisely, this plot was generated by first aggregating all tree pairs corresponding to the same hybridization number and then by computing the mean average runtime of each of those aggregated subsets. Therefore, each aborted tree pair whose computation would have taken longer than the given time limit of $20$~minutes, was treated as follows. If the tree pair could not be computed by both considered algorithms, this tree pair was not taken into account. However, if at least one algorithm was able to compute the hybridization number for a tree pair, the runtime of the other algorithm was set to $20$~minutes if this algorithm was aborted in this case. Regarding the Figure~\ref{14-fig-Hsim1}, the number assigned to each measurement denotes the number of tree pairs whose hybridization number could be computed within the time limit by the corresponding two algorithm. 

Figure~\ref{14-fig-Hsim1} indicates that each of our two modified algorithms outperforms the original algorithm. Moreover, this figure indicates that our third modified algorithm \scAMc~is significantly more efficient than our first modified algorithm \scAMa.

\begin{figure}
\centering
\includegraphics[scale=0.4]{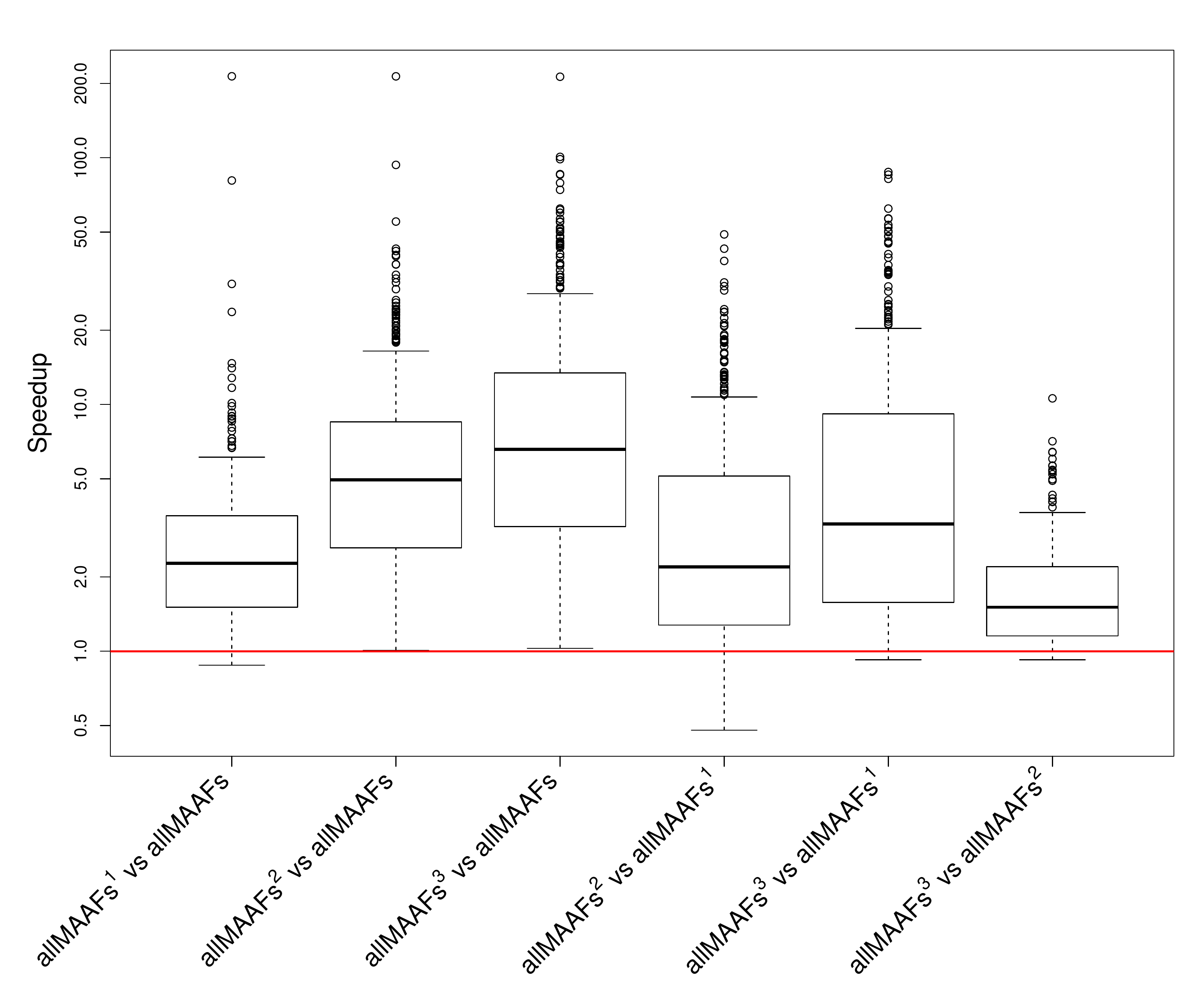}
\caption{The distributions of the speedups, obtained from comparing algorithm $A_1$ versus algorithm $A_2$, as denoted at the x-axis, via boxplots. A tree pair $d$ from our synthetic dataset was only considered if at least one of the two algorithms could process $d$ within the time limit of $20$ minutes and if at least one algorithm took longer than $50$ seconds for its computation. Notice that a base-10 log scale is used for the y-axis.} 
\label{14-fig-speedUp}
\end{figure}

Figure~\ref{14-fig-speedUp} shows the distributions of the speedups obtained from our two modified algorithms via boxplots. Each of those boxplots was generated by first selecting a relevant set $\cD$ of tree pairs from our synthetic dataset and then by computing the speedup of each of those tree pairs by taking the runtime obtained from the corresponding two algorithms into account. More precisely, in a first step, we set the runtime of each tree pair whose hybridization number could not be computed within the time limit to $20$~minutes. Second, we consider each tree pair $d$ as relevant if at least one of both algorithms could process $d$ within the time limit of $20$ minutes and if at least one computation referring to one of both algorithms took longer than $50$ seconds. Consequently, the relevant set $\cD$ excludes those tree pairs whose computational complexity is, on the one hand, too low and, on the other hand, too high to reveal a difference between the runtimes of both algorithms. 

Figure~\ref{14-fig-speedUp} again indicates that both modified algorithms are more efficient than the original algorithm. More specifically, for the considered set of tree pairs, \scAMa~and \scAMc~is on mean average about 3.9 times and 11.9 times (median 2.3, 4.9 and 6.6), respectively, faster than \textsc{allMAAFs}. Moreover, our second modified algorithm can significantly improve the practical runtime of our first modified algorithm. More specifically, for the considered set of tree pairs, \scAMc~is on mean average 8.15 times (median 3.18), respectively, faster than \scAMa.

\begin{figure}
\centering
\includegraphics[scale=0.4]{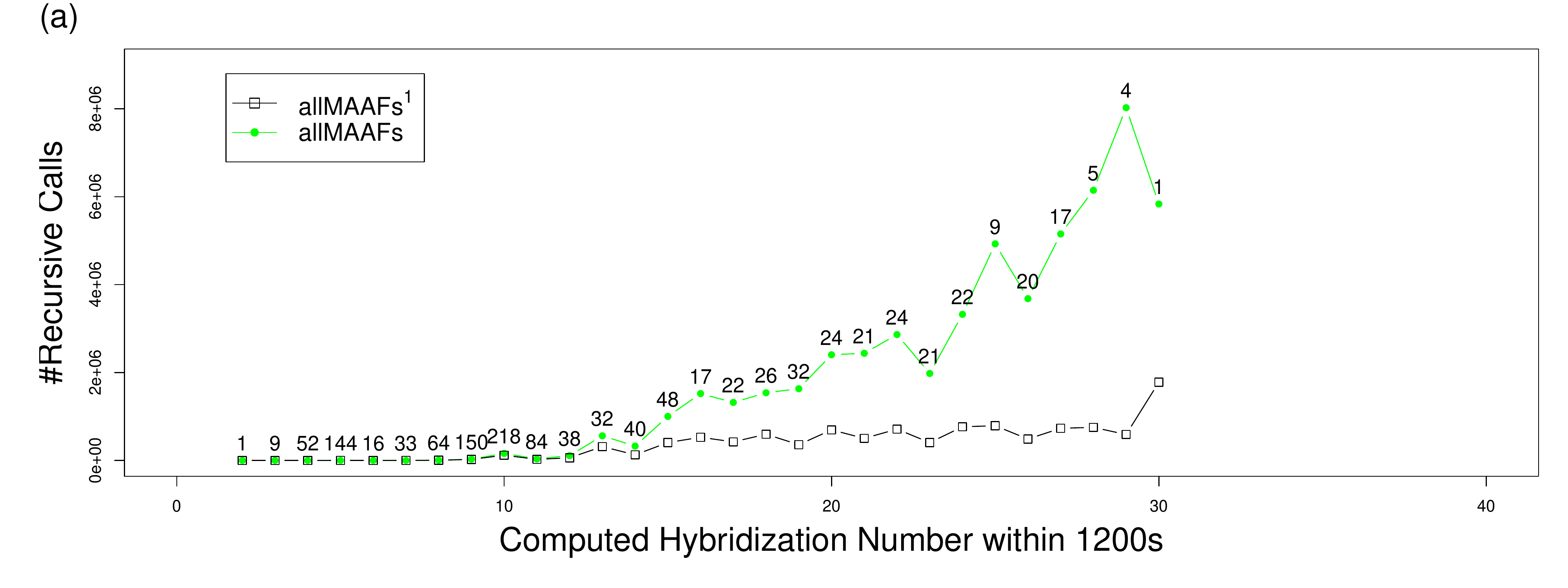}
\includegraphics[scale=0.4]{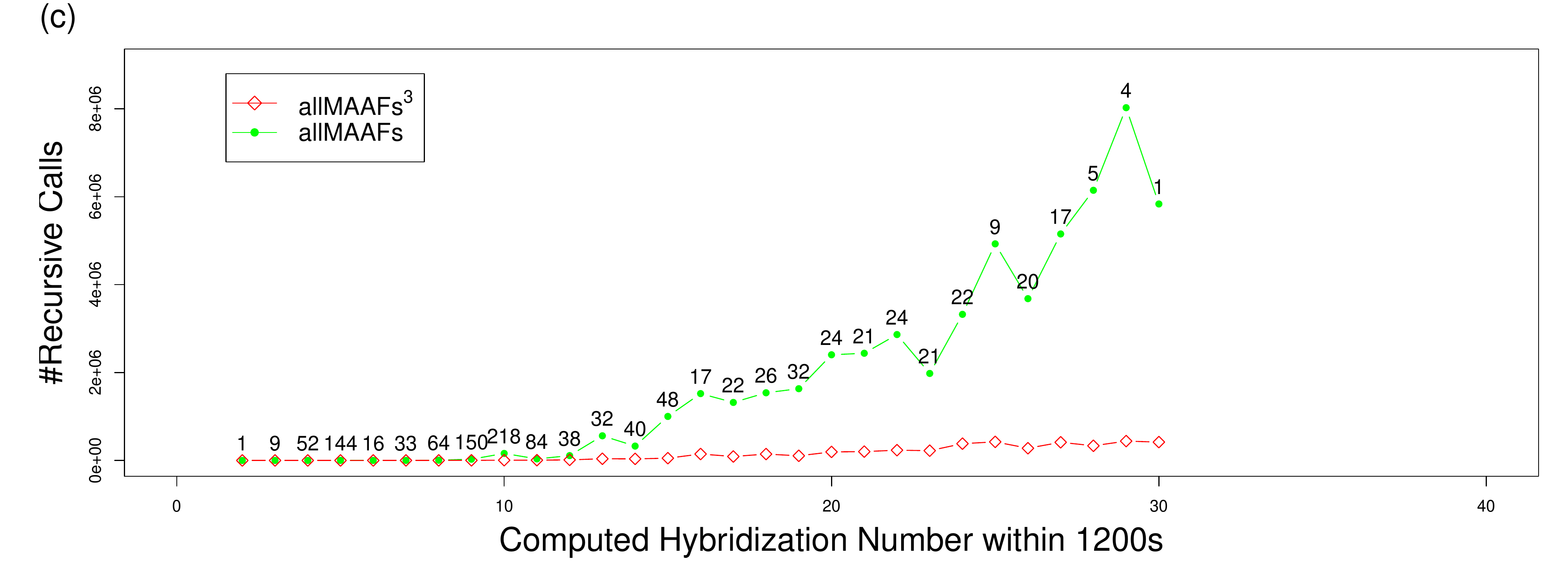}
\includegraphics[scale=0.4]{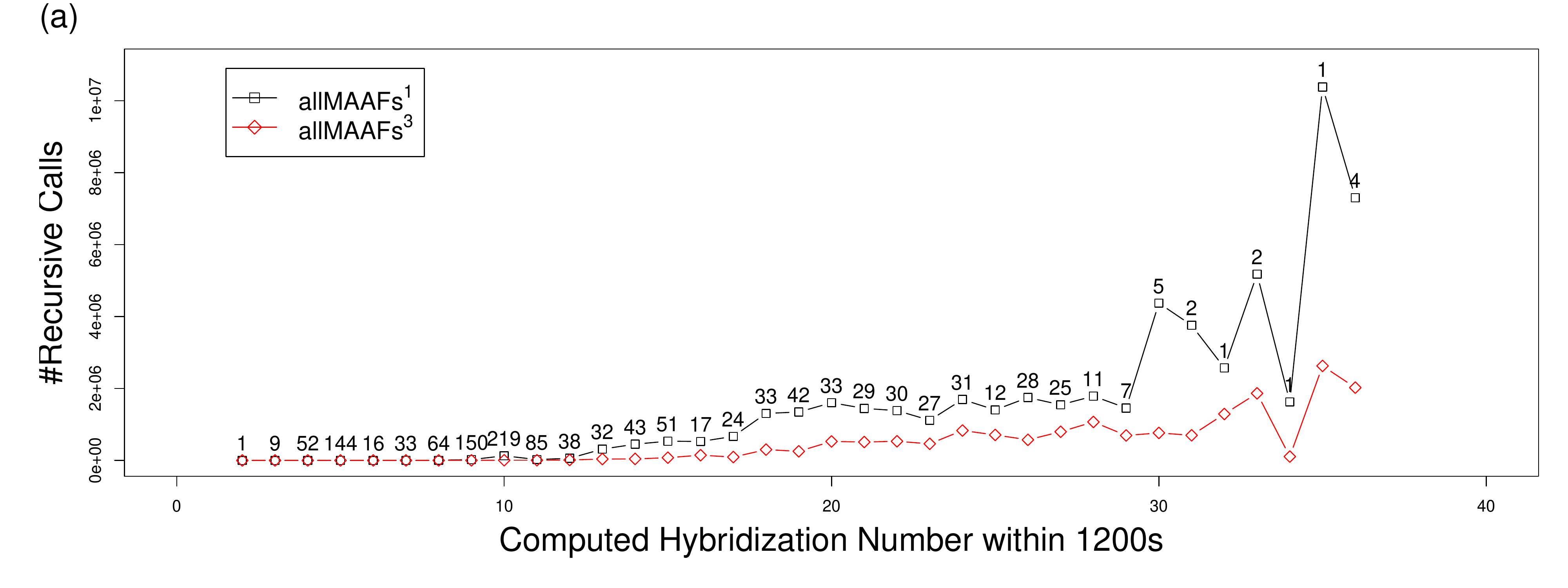}
\caption{Comparisons between the mean average number of recursive calls in terms of the hybridization number of (a) \textsc{allMAAFs} and \scAMa~as well as (c) \textsc{allMAAFs} and \scAMc. For each hybridization number the corresponding number of tree pairs is given that could be computed within the time limit of $20$~minutes by the two respective algorithms and, thus, contributed to the mean average number of recursive calls denoted by the y-axis.} 
\label{14-fig-calls1}
\end{figure}

In order to give a reason of the speedup obtained by our two modified algorithms, we measured the number of recursive calls that have been performed for calculating the hybridization number of each tree pair. To draw comparisons, we only took those tree pairs into account whose hybridization number could be processed by both algorithms within the time limit of $20$~minutes, which were in total $1302$. Next, all these tree pairs were grouped according to their hybridization number and, finally, the mean average number of recursive calls for each group was computed. Regarding Figure~\ref{14-fig-calls1}, the numbers that are attached to each measurement correspond to the number of tree pairs that have contributed to the mean average denoted at the x-axis. 

Figure~\ref{14-fig-calls1} indicates that by applying our two modified algorithms there are significantly less recursive calls necessary for the computation of a maximum acyclic agreement forest compared with the original algorithm. Moreover, Figure~\ref{14-fig-calls1}~(c) shows that our second modified algorithm \scAMc~has to conduct significantly less recursive calls than our first modified algorithm for computing hybridization numbers which, obviously, compensates the effort of preventing additional recursive calls when processing a common cherry. Since the computation of the hybridization number is just an intermediate step in computing all maximum acyclic agreement forests, the difference between the number of recursive calls between the first and both the second and the third modified algorithm is expected to be even larger in this case.

\begin{figure}
\centering
\begin{tabular}{cc}
\includegraphics[width = 5.5cm]{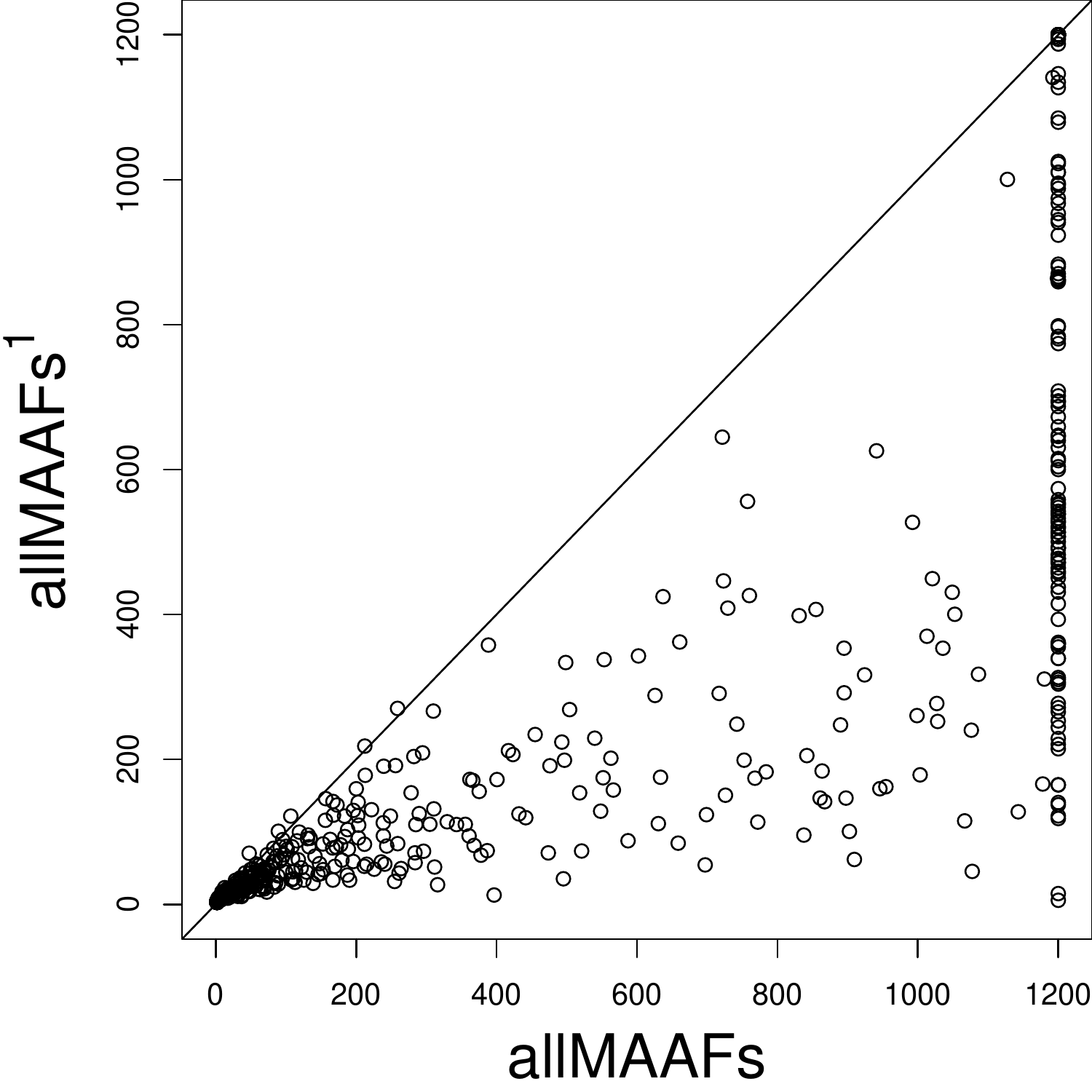}\hspace{0.2cm}
&
\includegraphics[width = 5.5cm]{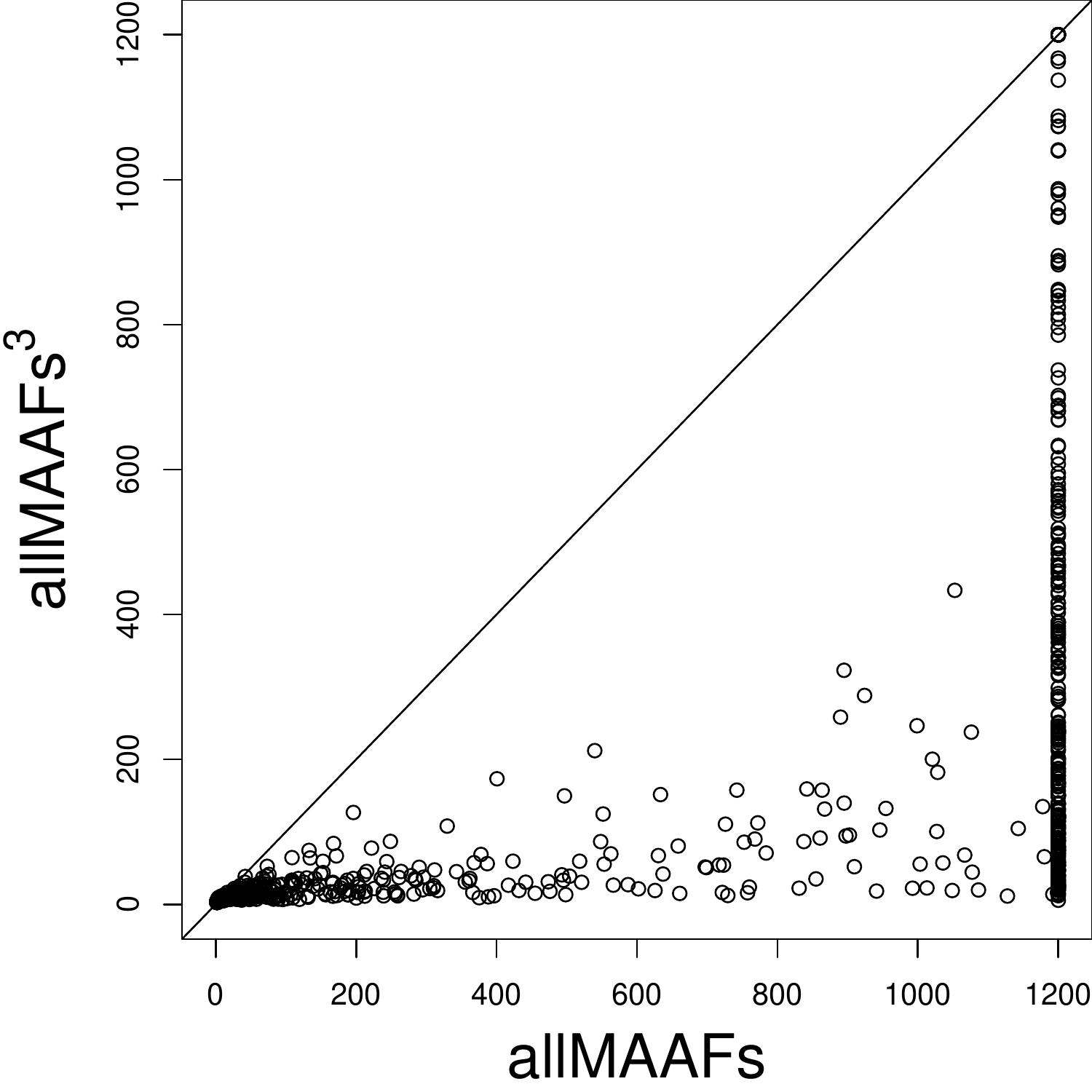}
\end{tabular}
\includegraphics[width = 5.5cm]{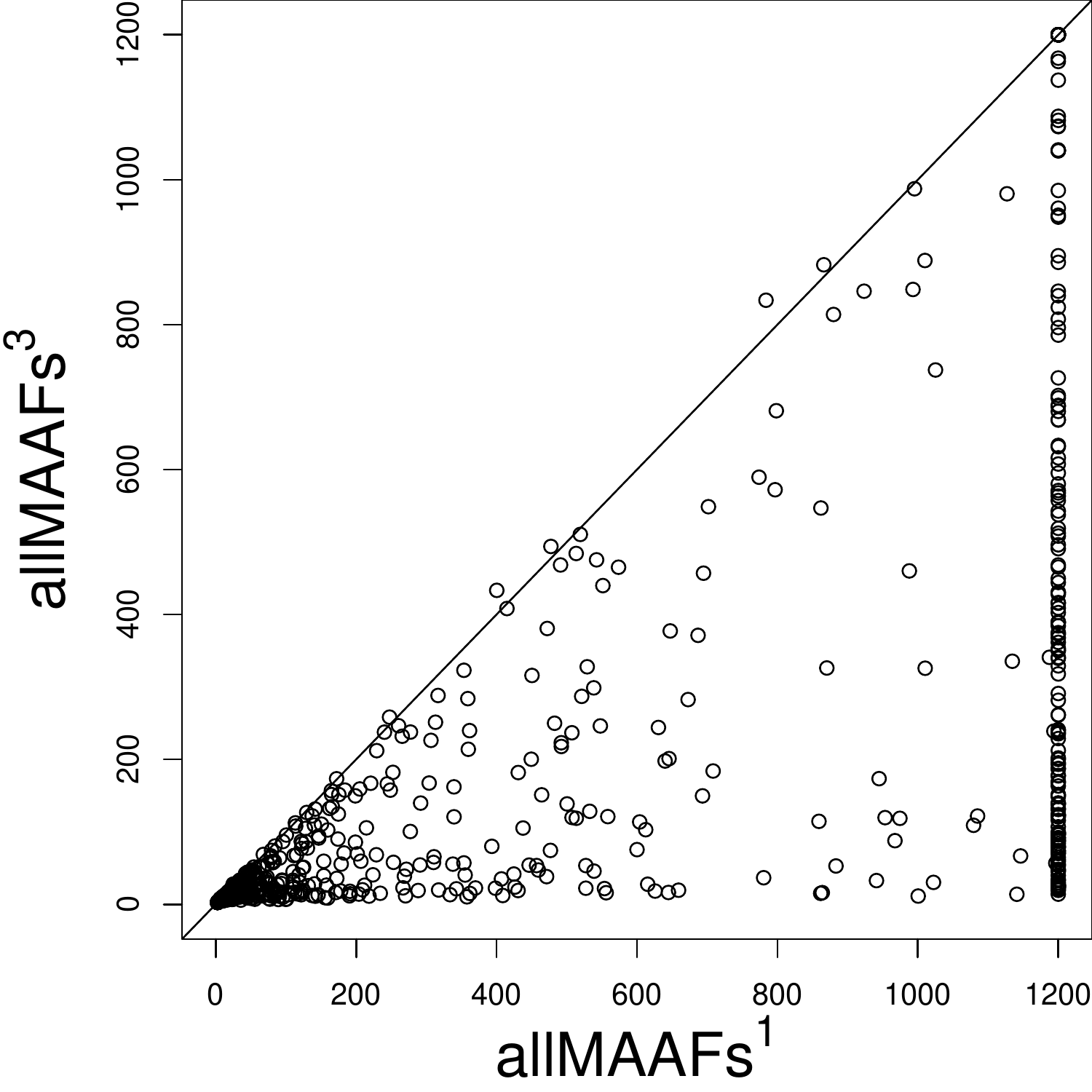}
\caption{Scatter-plots comparing the runtimes produced by the three different algorithms when applying our synthetic dataset. Each dot refers to a tree set of this data set whose corresponding x- and y-value indicates the runtime attained by the respective two algorithms.} 
\label{14-fig-scat}
\end{figure}

Finally, we finish this section by comparing the runtimes of the four algorithms via scatter-plots (cf.~Fig.~\ref{14-fig-scat}). Again, all tree pairs whose hybridization number could not be computed within the time limit was set to $20$~minutes. The plots clearly show that our modified algorithms outperform the original algorithm whereat the modified algorithm \scAMc~is more efficient than the modified algorithm \scAMa.

\section{Running time of \scAMa~and \scAMc}
\label{sec-rt}

The theoretical worst-case runtime of \scAMa~is still the same as the one of the original algorithm \textsc{allMAAFs}, which is $O(3^{|\cX|}|\cX|)$ as stated in the work of Scornavacca \emph{et al.}~\cite[Theorem 3]{Scornavacca2012}. 

\begin{theorem}
Let $T_1$ and $T_2$ be two rooted phylogenetic $\cX$-trees and let $k$ be an integer. The algorithm \scAMa$(T_1,T_2,T_1,\{T_2\},k)$ has a worst-case running time of $O(3^{|\cX|}|\cX|)$.
\label{14-th-rt1}
\end{theorem}

\begin{proof}
Let $\cF=\{F_{\rho},F_1,F_2,\dots,F_{k-1}\}$ be an agreement forest for $T_1$ and $T_2$ of size $k$. To obtain $\cF$ from $T_2$, one obviously has to cut $k-1$ edges. Moreover, in order to reduce the size of leaf set $\cX$ of $R$ to $1$, to each component $F_i$ in $\cF$ we have to contract exactly $|\cL(F_i)|-1$ cherries. Consequently, at most $|\cX|$ cherry contractions have to be performed in total. Thus, our algorithm has to perform at most $k+|\cX|=O(|\cX|)$ recursive calls for the computation of $\cF$. Now, as one of these recursive calls can at least branch into three further recursive calls, $O(3^{|\cX|})$ is an upper bound for the total number of recursive calls that are performed throughout the whole algorithm. Moreover, as each operation that is performed during a recursive call can be performed in $O(|\cX|)$ time, the worse-case running time of both algorithms is $O(3^{|\cX|}|\cX|)$. 
\end{proof}

When considering the theoretical worst-case runtime of \scAMc, we have to take the running time of a refinement step into account.

\begin{theorem}
Let $T_1$ and $T_2$ be two rooted phylogenetic $\cX$-trees and let $k$ be an integer. The algorithm \scAMc$(T_1,T_2,T_1,\{T_2\},k)$ has a worst-case running time of $O(3^{|\cX|}4^k|\cX|)$.
\label{14-th-rt2}
\end{theorem}

\begin{proof}
As stated in Theorem~\ref{14-th-rt1}, the algorithm has to conduct $O(3^{|\cX|})$ recursive calls. Potentially, for each of these recursive calls we have to apply a refinement step whose theoretical running time is stated with $O(4^k|\cX|)$ \cite{Whidden2011}. Moreover, as all other operations that are performed during a recursive call can be performed in constant time, the worse-case running time of the algorithm is $O(3^{|\cX|}4^k|\cX|)$.
\end{proof}

Even though both presented modifications do not improve the theoretical worst-case running time, our simulation study indicates that in practice these modifications are significantly faster than the original algorithm \textsc{allMAAFs}. Regarding our first modification \scAMa, this is simply due to the fact that the number of computational paths arising from processing a contradicting cherry is reduced. More precisely, if there is a computational path calculating a maximum acyclic agreement forest in which the set of pendant edges $E_B$ for a particular cherry is cut, the original algorithms produces $|E_B|-1$ redundant recursive calls.

Due to the refinement steps that are performed at the end of each recursive call, our third modification \scAMc~always has to start only one recursive call for processing a common cherry. This refinement step, however, has to be conducted for each maximum agreement forest $\cF$ in order to transform $\cF$ into all maximum acyclic agreement forests. Nevertheless, as indicated by our simulation study, this post processing step is still more efficient than always running two additional recursive calls when processing a common cherry (as it is the case for the original algorithm \textsc{allMAAFs}).

In conclusion, due to the significant speedup that can be obtained from applying both presented modifications, we feel certain that this work describes a noticeable improvement to the original algorithm \textsc{allMAAFs}, which will make this algorithm accessible for a wider range of biological real-world applications.

\clearpage
\section{Conclusion}
\label{sec-disc}

An important approach to study reticulate evolution is the reconciliation of incongruent gene trees into certain kinds of phylogenetic networks, so-called hybridization networks. From a biological point of view, it is important to calculate not only one but \emph{all} of those networks as, once calculated all of these networks, one can then apply certain filtering techniques testing certain hypothesis. As already discussed in \citet{Scornavacca2012}, the recently published algorithm \textsc{allMAAFs} calculating all maximum acyclic agreement forests for two rooted binary phylogenetic $\cX$-trees, is considered to be an important step to achieve this goal.

In this paper, we have presented the two modifications \scAMa~and \scAMc~of the algorithm \textsc{allMAAFs} still calculating all maximum acyclic agreement forests for both input trees as shown by formal proofs. Moreover, as we have additionally indicated by a simulation study that both modifications significantly improve the practical running time of the original algorithm, we feel certain that this work will facilitate the computation of hybridization networks for larger input trees and, thus, makes this approach accessible to a wider range of biological problems.

\section*{Acknowledgements}
We gratefully acknowledge Simone Linz for useful discussions.

%\bibliographystyle{spbasic} % basic style, author-year citations
%\bibliographystyle{spmpsci} % mathematics and physical sciences
%\bibliographystyle{spphys} % APS-like style for physics
%\bibliography{} % name your BibTeX data base

\begin{thebibliography}{}

\bibitem[Albrecht~\textit{et.~al.}, 2011]{Albrecht2011}
Albrecht, B.; Scornavacca, C.; Cenci, A.; Huson, D. H.: 
Fast computation of minimum hybridization networks. 
\textit{Bioinformatics}, 28(2): 191--197, 2011. 

\bibitem[Albrecht, 2014]{Albrecht2014}
Albrecht, B.:
Computing Hybridization Networks for Multiple Rooted Binary Phylogenetic Trees by Maximum Acyclic Agreement Forests.
arXiv:1408.3044, 2014.

\bibitem[Baroni~\textit{et.~al.}, 2005]{Baroni2005}
Baroni, M.; Gruenewald, S.; Moulton, V.; Semple, C.: 
Bounding the number of hybridisation events for a consisten evolutionary history. 
\textit{Mathematical Biology} 51: 171--182, 2005.

\bibitem[Baroni~\textit{et.~al.}, 2006]{Baroni2006}
Baroni, M.; Semple, C.; Steel, M.: 
Hybrids in real time. 
\textit{Systematic Biology}, 55: 46--56, 2006.

\bibitem[Bordewich~and~Semple, 2007]{Bordewich-2007/1}
Bordewich, M.; Semple, C.: 
Computing the minimum number of hybridization events for a consistent evolutionary history. 
\textit{Discrete Applied Mathematics}, 155: 914--928, 2007.

\bibitem[Bordewich~and~Semple, 2007]{Bordewich2007}
Bordewich, M.; Semple, C.: 
Computing the Hybridization Number of Two Phylogenetic Trees Is Fixed-Parameter Tractable. 
\textit{IEEE/ACM Trans. Comput. Biol. Bioinformatics}, 4(3): 458--466, 2007.

\bibitem[Chen~and~Wang, 2010]{Chen2010}
Chen, Z. Z.; Wang, L.: 
HybridNET: a tool for constructing hybridization networks. 
\textit{Bioinformatics}, 26: 2912--2913, 2010.

\bibitem[Collins~\textit{et.~al.}, 2011]{Collins2011} 
Collins, J.; Linz, S.; Semple, C: 
Quantifying hybridization in realistic time.
\textit{Journal of Computational Biology}, 18(10): 1305--1318, 2011.

\bibitem[Huson~\textit{et.~al.}, 2007]{Huson2007}
Huson, D.; Rupp, R.; Scornavacca, C.: 
Phylogenetic Networks: Concepts, Algorithms and Applications. 
\textit{Cambridge University Press}, 2011.

\bibitem[Huson~\textit{et.~al.}, 2011]{Huson2011}
Huson, D.; Richter, D.; Rausch C.; Dezulian T.; Franz M.; Rupp R.: 
Dendroscope: An interactive viewer for large phylogenetic trees. 
\textit{BMC Bioinformatics}, 22;8(1):460, 2007.

\bibitem[Mallet, 2007]{Mallet-2007}
Mallet,J.:
Hybrid speciation. 
\textit{Nature}, \textbf{446}, 279-–283, 2007.

\bibitem[Scornavacca~\textit{et.~al.}, 2012]{Scornavacca2012}
Scornavacca, C.; Simone, L.; Albrecht, B.: 
A first step towards computing all hybridization networks for two rooted binary phylogenetic trees.
\textit{Journal of Computational Biology}, 2012. 

\bibitem[Whidden~and~Zeh, 2009]{Whidden2009} 
Whidden, C; Zeh, N.: 
A unifying view on approximation and fpt of agreement forests. 
\textit{Algorithms in Bioinformatics}, LNCS Vol.\,5724, 390--402, 2009. 

\bibitem[Whidden~\textit{et.~al.}, 2011]{Whidden2011} 
Whidden, C.; Beiko, R.; Zeh, N.: 
Fixed-Parameter and Approximation Algorithms for Maximum Agreement Forests. 
arXiv:1108.2664, 2011.

\end{thebibliography}


\end{document}